\documentclass[twocolumn]{autart}  
\makeatletter
\let\ps@copyright\ps@plain
\makeatother
\usepackage{amsmath,amssymb,mathtools,bm}
\usepackage{graphicx}
\usepackage{enumitem}
\usepackage{algorithm}

\usepackage{algorithmic}
\usepackage{xcolor}
\usepackage{comment}
\usepackage{tikz}
\usepackage{cite}
\usepackage{url}
\usepackage{booktabs}
\usepackage{subcaption}

\newtheorem{assumption}{Assumption}
\newtheorem{definition}{Definition}

\newtheorem{proposition}{Proposition}
\newtheorem{theorem}{Theorem}

\newtheorem{corollary}{Corollary}

\newcommand{\R}{\mathbb{R}}

\newcommand{\cV}{\mathcal{V}}

\newcommand{\cA}{\mathcal{A}}

\newcommand{\cE}{\mathcal{E}}
\newcommand{\cG}{\mathcal{G}}
\newcommand{\cN}{\mathcal{N}}
\newcommand{\evec}{\mathbf{e}}
\newcommand{\F}{\mathbf{F}}
\newcommand{\K}{\mathbf{K}}

\newcommand{\Fb}{\boldsymbol{\mathcal{F}}}

\newcommand{\col}{\operatorname{col}}
\newcommand{\diag}{\operatorname{diag}}
\newcommand{\blkdiag}{\operatorname{blkdiag}}

\newcommand{\x}{\mathbf{x}}
\newcommand{\uvec}{\mathbf{u}}
\newcommand{\w}{\mathbf{w}}

\newcommand{\Erre}{\mathbf{R}} 

\begin{document}

\begin{frontmatter}

\title{Distributed Performance Boosting through Neural Control for Plug-and-Play Networks} 

\thanks{This work was supported as a part of NCCR Automation, a National Centre of Competence in Research, funded by the Swiss National Science Foundation (grant number 51NF40\_225155), and the NECON project (grant number 200021-219431). Implementation
of the proposed method is available at \url{https://github.com/DaniloSaccani/PBdistributed_PnP}. }          

\author[EPFL]{Danilo Saccani}\ead{danilo.saccani@epfl.ch}, 
\author[EPFL]{Leonardo Massai}\ead{l.massai@epfl.ch}, 
\author[EPFL]{Giancarlo Ferrari-Trecate}\ead{giancarlo.ferraritrecate@epfl.ch},               

\address[EPFL]{Ecole polytechnique federale de Lausanne (EPFL), CH-1015 Lausanne, Switzerland.}  %

\begin{abstract}
We address the problem of improving the performance of pre-stabilized
networked systems subject to plug-and-play (PnP) operations, namely
additions or removals of subsystems. While distributed baseline
controllers can preserve stability after a PnP transition, augmenting
them with expressive neural policies can invalidate the
formal stability guarantees. To overcome this issue, we build on recently proposed Performance Boosting controllers based on the Internal Model Control (IMC) principle, which reduce
closed-loop stability to a finite $\ell_2$-gain requirement on an
auxiliary operator. 
In the distributed setting, this operator is formed by interconnecting local neural operators that exchange information only between neighboring subsystems. We then derive \emph{free parametrizations} of this interconnection, namely explicit maps from unconstrained trainable variables to interconnected neural operators satisfying a prescribed network gain bound for every parameter value. The obtained distributed neural controller can therefore be trained by unconstrained gradient-based optimization without stability constraints or projection steps, while retaining the prescribed gain bound at every training iterate.
We develop two complementary parametrizations. A neighbor-based compositional parametrization uses local coupling bounds, so that after a PnP event only subsystems whose bounds change need to update their gain parameters. A network-norm parametrization instead uses a bound on the overall coupling matrix and can reduce conservatism, at the price of requiring network-wide information.
A PnP DC-microgrid case study shows that the proposed neural layer improves
voltage transients under topology changes and load variations while
preserving nominal closed-loop stability.
\end{abstract}

\end{frontmatter}

\section{Introduction}
Modern large-scale cyber-physical systems, such as smart grids,
vehicular platoons, and multi-agent robotic networks, rely fundamentally on distributed control strategies. In such systems, decisions are made
by individual subsystems using limited information, avoiding the
communication burden, latency, and fragility of centralized
architectures. At the same time, dynamic couplings imply that poorly designed individual actions may propagate through the network and destabilize the overall interconnection~\cite{antonelli2013interconnected}.
This makes stability guarantees a central requirement in distributed control design.

Model-based distributed control provides several tools to certify the
stability of interconnected systems. Distributed Model Predictive
Control (MPC) schemes enforce stability and constraints through
terminal ingredients, invariant sets, or tube-based
arguments~\cite{aboudonia2022reconfigurable,carron2023multi,saccani2023model,zeilinger2013plug}.
Alternatively, dissipativity and passivity-based methods provide
compositional certificates for network stability, often leading to subsystem-wise linear matrix inequality conditions or explicit inequalities that depend only on each subsystem model, its controller and neighboring coupling terms~\cite{cucuzzella2019distributed,aboudonia2021passivity,fu2024distributed,malan2024passivity}. 
These approaches provide stability by design, which is essential in safety-critical networks. However, they typically optimize restricted performance criteria, such as quadratic costs or robustness margins, or derive stability conditions without optimizing a task-dependent performance objective. Thus, a key challenge remains: how to retain formal closed-loop stability guarantees while optimizing expressive distributed feedback laws for general, possibly non-convex performance objectives.

Learning-based methods, such as Reinforcement Learning (RL), address the
performance side of this challenge by optimizing highly expressive
feedback laws directly from data~\cite{govinda2025survey}.
In networked systems, multi-agent RL and Graph Neural Networks (GNNs)
provide scalable policy classes based on subsystem measurements and
neighbor-to-neighbor messages, which can promote generalization across
network topologies~\cite{zhang2021multi}.
Yet, these approaches typically do not provide
hard closed-loop stability guarantees; stability is instead encouraged
indirectly through reward shaping, regularization, or external safety
layers~\cite{garcia2015comprehensive}. Moreover, in these approaches, scalability often relies on shared policy parameters and homogeneous subsystem models, which can be restrictive in heterogeneous engineering networks.

A recent line of work addresses stability by parametrizing classes of
stabilizing controllers through nonlinear extensions of Youla, Internal
Model Control (IMC), and system-level
parameterizations~\cite{manchester2026neural,furieri2022neural,galimberti2025parametrizations}.
These parametrizations have enabled neural controllers to be trained on
top of pre-stabilized nonlinear systems by minimizing general, possibly
non-convex performance objectives while preserving closed-loop
stability~\cite{furieri2024learning,galimberti2025parametrizations}.
To train such controllers, one needs a finite-dimensional parametrization of a suitable subclass of stabilizing controllers. This motivates \emph{free parametrizations}: explicit maps from unconstrained finite-dimensional variables to dynamic operators satisfying the required stability or gain condition by construction. Controller training can then be performed through unconstrained gradient-based optimization, without projections or explicit stability constraints. Several such constructions have been developed for centralized operators~\cite{manchester2026neural,revay2023recurrent,massai2025free}. 

For distributed controllers, local stability or gain guarantees are not sufficient. The interconnection must preserve the desired network-level property, while the controller must respect the communication graph. Previous work has addressed these two requirements in different ways. Free parametrizations of interconnected neural operators have been developed by composing local gain bounds into a network-level condition~\cite{massai2024unconstrained,saccani2024optimal}. The GNN-based construction of~\cite{cao2025scaling} instead uses a common graph-based policy across the subsystems, allowing a policy trained on a smaller network to be deployed on larger networks and unseen topologies within the same problem class. Port-Hamiltonian approaches obtain compositional guarantees by exploiting power-preserving system interconnections or by imposing port-Hamiltonian structure on distributed controllers with prescribed gain bounds~\cite{furieri2022distributed,zakwan2026neural}.

In practice, networked control systems need not have a fixed configuration. Applications such as microgrids and multi-robot systems may combine heterogeneous subsystems whose composition and interconnection change over time. Plug-and-play (PnP) control addresses this setting by allowing subsystems to be added or removed while preserving closed-loop guarantees. Model-based PnP control has been studied through robust MPC and distributed invariance arguments~\cite{riverso2013plug,zeilinger2013plug,aboudonia2022reconfigurable,bodmer2026anytime}, as well as through dissipativity, passivity, and compositional certificates~\cite{tucci2017line,strehle2020scalable,nahata2020passivity}. These methods provide stability guarantees for baseline local controllers after network changes. However, when a neural feedback law is added to improve performance, the original PnP stability guarantee does not automatically extend to the augmented closed loop.

In this paper, we address this problem by designing distributed neural controllers for heterogeneous networked systems subject to PnP operations. We construct free parametrizations of an interconnected IMC operator that enforce a prescribed network gain for every value of the trainable variables, and characterize how this gain can be restored after admissible topology changes.

\noindent\emph{Contributions.}
The contributions are threefold.
(i) We derive two free parametrizations of interconnected neural
operators satisfying a prescribed network $\ell_2$ gain. The
\emph{neighbor-based} construction uses coupling information associated
with each subsystem and its neighbors, whereas the \emph{network-norm}
construction uses a bound on the coupling matrix of the overall network
and can reduce conservatism at the price of requiring network-wide information.
(ii) We characterize post-PnP recertification. For the neighbor-based parametrization, we give conditions under which only subsystems whose local coupling bounds change need to update their gain parameters. The network-norm parametrization instead generally requires updating a network-wide bound, unless a common bound for a known family of admissible PnP configurations is available in advance.
(iii) We show that, under strictly causal local realizations and separable finite-horizon costs, controller training can be implemented through neighbor-to-neighbor communication while preserving the prescribed network gain at every training iterate.

In~\cite{saccani2024optimal}, a preliminary version of this work, (i) only a neighbor-based free parametrization for a fixed network topology was considered, (ii) PnP operations and post-change recertification were not addressed, (iii) a neighbor-to-neighbor implementation of controller training was not established, and (iv) the detailed proofs of the main results were skipped.

The paper is organized as follows.
Section~\ref{sec:prelim} introduces the distributed system, the PnP
setting, and the control problems.
Section~\ref{sec:proposed} presents the distributed neural architecture,
the two free parametrizations, and the PnP recertification result.
Section~\ref{sec:disttrain} shows how training can be performed through
neighbor-to-neighbor computations and discusses post-PnP adaptation.
Section~\ref{sec:num} reports the numerical experiments, and
Section~\ref{sec:conclusions} concludes the paper.

\subsection*{Notation}
Let $\ell^n$ denote the space of sequences
$\mathbf a=(a_0,a_1,\ldots)$ with $a_t\in\R^n$.
Bold symbols denote sequences, with
$\|\mathbf a\|^2:=\sum_{t=0}^{\infty}\|a_t\|^2$, where $\|a_t\|$ denotes the Euclidean norm, and
$\ell_2^n:=\{\mathbf a\in\ell^n:\|\mathbf a\|<\infty\}$.
For matrices, $\|A\|_2$ denotes the spectral norm.
For an operator $\mathbf A:\ell^n\to\ell^m$, $(\mathbf A(\mathbf a))_t\in\R^m$ denotes its output at time $t$.
The operator is \emph{causal} if $(\mathbf A(\mathbf a))_t$ depends only on $a_{0:t}$.
We write $\mathbf A\in\mathcal L_2$ if $\mathbf A$ is causal and maps $\ell_2^n$ to $\ell_2^m$, and call it finite-gain $\ell_2$-stable if there exist $\gamma,\beta\ge0$ such that $\|\mathbf A(\mathbf a)\|\le\gamma\|\mathbf a\|+\beta$ for all $\mathbf a\in\ell_2^n$.
A function $\alpha:\R_{\ge0}\to\R_{\ge0}$ belongs to $\mathcal K$ if it is continuous, strictly increasing, and satisfies $\alpha(0)=0$.
For vectors $a_i$, $\col_i(a_i)$ denotes their column-wise stacking, and $\blkdiag(A_1,\ldots,A_N)$ denotes the block-diagonal matrix or operator with diagonal blocks $A_1,\ldots,A_N$.

\section{Plug-and-play networks and problem formulation}\label{sec:prelim}
We consider $N$ nonlinear time-varying subsystems coupled over an
undirected connected graph $\cG=(\cV,\cE)$, with
$\cV=\{1,\ldots,N\}$. For each $i\in\cV$, let
$\cN_i:=\{i\}\cup\{j\in\cV:\{i,j\}\in\cE\}$ denote its neighborhood.
The graph also specifies the allowed neighbor-to-neighbor communication.
The dynamics of subsystem $i$ are
\begin{equation}\label{eq:nonlinsyslocal}
x_{t+1}^{[i]}
=
f_t^{[i]}\!\left(x_t^{[\cN_i]},u_t^{[i]}\right)
+
w_t^{[i]},
\qquad t\ge0,
\end{equation}
where $x_t^{[i]}\in\R^{n_i}$, $u_t^{[i]}\in\R^{m_i}$, and
$w_t^{[i]}\in\R^{n_i}$. 
Throughout the paper, \textit{local} refers to quantities associated
with a single subsystem or its local controller, while \textit{neighbor-to-neighbor}
refers explicitly to information exchanged over the graph
$\mathcal G$ in one hop.
By stacking all local variables as
$x_t=\col_{i\in\cV}(x_t^{[i]})$, $u_t=\col_{i\in\cV}(u_t^{[i]})$,
and $w_t=\col_{i\in\cV}(w_t^{[i]})$, the network dynamics read
\begin{equation}\label{eq:nonlinsys}
x_{t+1}=f_t(x_t,u_t)+w_t .
\end{equation}
Let $\F:\ell^{n}\times\ell^{m}\to\ell^{n}$ be the strictly causal
operator induced by the sequence $\{f_t\}_{t\ge0}$, namely
\[
(\F(\x,\uvec))_0=0,\qquad
(\F(\x,\uvec))_{t+1}=f_t(x_t,u_t).
\]
For a given initial condition $x_0$, define
$\evec:=\mathbf E(x_0)+\mathbf S\w$, where
$(\mathbf E(x_0))_0=x_0$, $(\mathbf E(x_0))_{t>0}=0$,
$(\mathbf S\w)_0=0$, and $(\mathbf S\w)_{t+1}=w_t$, $t\geq0$. Then
\eqref{eq:nonlinsys} can be written compactly as
\begin{equation}\label{eq:plant-operator}
\x=\F(\x,\uvec)+\evec .
\end{equation}
When the physical system has a nonzero operating equilibrium, we
interpret $x_t$ as the deviation from the equilibrium associated with
the current network configuration. Hence, the origin in the following
stability statements represents this operating equilibrium.
The corresponding local operators $\F^{[i]}$ are induced by the sequences $\{f_t^{[i]}\}_{t\ge0}$, so that
\begin{equation}\label{eq:sub-operator}
\x^{[i]}
=
\F^{[i]}(\x^{[\cN_i]},\uvec^{[i]})
+
\mathbf E^{[i]}(x_0^{[i]})
+
\mathbf S\w^{[i]} .
\end{equation}
For every initial condition $x_0$ and every input-disturbance pair $(\uvec,\w)\in\ell_2^m\times\ell_2^n$, the equation~\eqref{eq:plant-operator} admits a unique causal state trajectory $\x$. Thus, the plant induces the unique transition operator $\Fb_{x_0}:(\uvec,\w)\mapsto \x$. 

A PnP operation models the physical addition or removal of subsystems at time $t_{\mathrm p}$. We represent such an event as a transition from a pre-change configuration $(\cG,\F)$ to a post-change configuration $(\cG^+,\F^+)$. Here, $\cG^+=(\cV^+,\cE^+)$ denotes the post-change communication graph, while $\F^+$ is the operator induced by the post-change dynamics
$x_{t+1}=f_t^+(x_t,u_t)+w_t$, $t\ge t_{\mathrm p}$.
Let $x_{t_{\mathrm p}}^+$ denote the post-change initial state obtained
by expressing the physical state at $t_{\mathrm p}$ in coordinates
centered at the post-change equilibrium. We denote by
$\Fb^+_{x_{t_{\mathrm p}}^+}:(\uvec,\w)\mapsto\x$ the corresponding
causal post-change plant map.
Standard PnP designs typically include an admission mechanism that
determines whether a topology change can be accepted by the baseline
controller. We call the PnP operations accepted by this mechanism
\emph{admissible} and make the following assumption on the resulting
pre-stabilized plant maps.

\begin{assumption}[Baseline PnP stability]\label{ass1:prestabilization} 
The baseline physical network is pre-stabilized. Specifically, for every initial condition $x_0$, the pre-change plant operator $\Fb_{x_0}$ is $\ell_2$-stable. Furthermore, every admissible PnP operation preserves this baseline stability: for any post-change pair $(\cG^+,\F^+)$ and any $x_{t_{\mathrm p}}^+$,
the post-change map $\Fb^+_{x_{t_{\mathrm p}}^+}$ is $\ell_2$-stable.
\end{assumption}
\begin{assumption}[Disturbance sampling]\label{ass2:disturbance}
The initial condition $x_0$ and disturbance sequence $\w$ are sampled from distributions $\mathcal D_0$ and $\mathcal D$, respectively. These distributions are unknown, but can be sampled.
\end{assumption}

Assumption~\ref{ass1:prestabilization} establishes the core setting of this paper: the physical network is already equipped with distributed primary controllers designed to ensure stability, both nominally and after admissible PnP operations, as in standard model-based PnP control schemes~\cite{zeilinger2013plug,riverso2013plug,nahata2020passivity,sadabadi2020line,strehle2020passivity}. 
Hence, throughout the paper, ``plant'' strictly refers to the map seen by the boosting layer after all baseline PnP feedback controllers have been absorbed into the operator; Assumption~\ref{ass1:prestabilization} does not require that the uncontrolled physical network is stable.

\subsection{Distributed controllers and PnP constraints}

We consider dynamical causal policies of the form
$\uvec=\K(\x,\uvec)$, where the operator $\K$ is strictly causal in its second
argument. Since dynamic distributed controllers may exchange internal
controller variables in addition to physical measurements, we define
distributed implementability through the information exchanged between
subsystems.

A causal operator is said to be \emph{$\cG$-distributed} if it can be
implemented by local dynamical systems, one per subsystem, that use
local measurements and exchange internal signals only with neighboring
subsystems. Local internal states and neighbor-to-neighbor signals are
allowed, so information may propagate over multiple hops dynamically
over time, while direct communication between non-neighboring
subsystems is excluded.
Accordingly, a policy $\K$ is called $\cG$-distributed if its controller
operator is $\cG$-distributed in the above sense.
The set of distributed controllers preserving closed-loop stability is therefore
\begin{multline*}
\mathcal K_{\mathrm{dist}}^{\ell_2}(\cG,\F)
:=
\Big\{
\K \ \Big|\
\K\ \text{is $\cG$-distributed and }\\
(\x,\uvec)\in\ell_2^n\times\ell_2^m, \ \forall  x_0\in\mathbb R^n, \  \forall \w\in\ell_2^n,
\Big\}.
\end{multline*}
A key feature of PnP control is that after a topology change, controller
updates should be restricted to the subsystems directly affected by the
change.
\begin{definition}[PnP update constraints]\label{def:pnp_update}
Consider a PnP transition from $(\cG,\F)$ to $(\cG^+,\F^+)$. The set of post-change subsystems allowed to update their local controllers is
\begin{multline}\label{eq:affected_set}
\cA(\cG,\cG^+)
:=
(\cV^+\setminus\cV)
\ \cup \\
\big\{
i\in\cV\cap\cV^+
\,\big|\,
\cN_i^+\neq\cN_i
\ \text{or}\
\F^{+[i]}\not\equiv\F^{[i]}
\big\}.
\end{multline}
Given a pre-change controller $\K=\operatorname{col}_{i\in\cV}\K^{[i]}$, a post-change controller $\K^+=\operatorname{col}_{i\in\cV^+}\K^{+[i]}$ satisfies the PnP update constraints if
\begin{equation}\label{eq:pnp_update_constraint}
\K^{+[i]}\equiv \K^{[i]},
\qquad
\forall i\in(\cV\cap\cV^+)\setminus\cA(\cG,\cG^+).
\end{equation}
\end{definition}
Thus, local controller updates are permitted only for newly added subsystems, and for surviving subsystems that experience a change in either their local pre-stabilized dynamics or their neighborhood connections. All other subsystems must retain their existing local controllers strictly unchanged. Figure~\ref{fig:pnp_update} illustrates a PnP transition and its corresponding affected set.

\begin{figure}[t!]
    \centering
\tikzset{every picture/.style={line width=0.75pt}} 

\begin{tikzpicture}[x=0.60pt,y=0.60pt,yscale=-1,xscale=1]

\draw   (60,50) .. controls (60,38.4) and (69.4,29) .. (81,29) .. controls (92.6,29) and (102,38.4) .. (102,50) .. controls (102,61.6) and (92.6,71) .. (81,71) .. controls (69.4,71) and (60,61.6) .. (60,50) -- cycle ;
\draw   (170,50) .. controls (170,38.4) and (179.4,29) .. (191,29) .. controls (202.6,29) and (212,38.4) .. (212,50) .. controls (212,61.6) and (202.6,71) .. (191,71) .. controls (179.4,71) and (170,61.6) .. (170,50) -- cycle ;
\draw   (280,50) .. controls (280,38.4) and (289.4,29) .. (301,29) .. controls (312.6,29) and (322,38.4) .. (322,50) .. controls (322,61.6) and (312.6,71) .. (301,71) .. controls (289.4,71) and (280,61.6) .. (280,50) -- cycle ;
\draw   (169,140) .. controls (169,128.4) and (178.4,119) .. (190,119) .. controls (201.6,119) and (211,128.4) .. (211,140) .. controls (211,151.6) and (201.6,161) .. (190,161) .. controls (178.4,161) and (169,151.6) .. (169,140) -- cycle ;
\draw    (102,50) -- (170,50) ;
\draw    (212,50) -- (280,50) ;
\draw    (190,119) -- (190,83) ;
\draw [shift={(190,80)}, rotate = 90] [fill={rgb, 255:red, 0; green, 0; blue, 0 }  ][line width=0.08]  [draw opacity=0] (10.72,-5.15) -- (0,0) -- (10.72,5.15) -- (7.12,0) -- cycle    ;
\draw    (239,31) .. controls (272,18) and (220,80) .. (250,70) ;

\draw (71,42.4) node [anchor=north west][inner sep=0.75pt]    {$\Sigma _{1}$};
\draw (182,41.4) node [anchor=north west][inner sep=0.75pt]    {$\Sigma _{2}$};
\draw (291,42.4) node [anchor=north west][inner sep=0.75pt]    {$\Sigma _{3}$};
\draw (181,132.4) node [anchor=north west][inner sep=0.75pt]    {$\Sigma _{4}$};

\end{tikzpicture}
    \caption{Example of a PnP update from $\cG$ to $\cG^+$. At time $t_{\mathrm p}$ subsystem
$\Sigma_3$ is removed, while $\Sigma_4$ is plugged in and connected to
$\Sigma_2$. The affected post-change set is
$\cA(\cG,\cG^+)=\{2,4\}$.}
    \label{fig:pnp_update}
\end{figure}
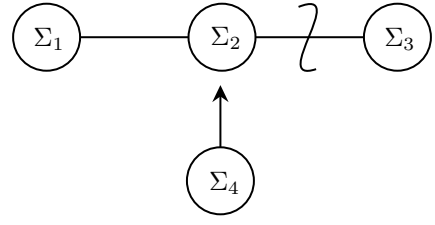

\subsection{Performance boosting and post-change update problems}
\label{subsec:problem_formulation}
The goal of performance boosting is to design an auxiliary policy that minimizes a task-dependent, possibly non-convex finite-horizon cost $J_T:\mathbb R^{n(T+1)}\times\mathbb R^{mT}\to\mathbb R$ without compromising the closed-loop stability guaranteed by the baseline controller. For a given $\F$ and $\K$, define the expected cost
\[
\bar J_T(\F,\K)
:=
\mathbb E_{x_0\sim\mathcal D_0,\,\w\sim\mathcal D}
\!\left[J_T(x_{0:T},u_{0:T-1})\right],
\]
where $(x_{0:T},u_{0:T-1})$ is the closed-loop trajectory generated by $x_0$, $w_{0:T-1}$, and $\K$. 

The performance boosting problem is
\begin{subequations}\label{eq:prob_design}
\begin{align}
\min_{\K}\quad
& \bar J_T(\F,\K) \label{eq:prob_design_obj}\\
\text{s.t.}\quad
& \K\in\mathcal K_{\mathrm{dist}}^{\ell_2}(\cG,\F).
\label{eq:prob_design_stab}
\end{align}
\end{subequations}
After a PnP operation producing $(\cG^+,\F^+)$, our primary objective
is instead to restore feasibility while modifying only the affected
controllers:
\begin{subequations}\label{eq:prob_pnp_update}
\begin{align}
\text{find}\quad & \K^+ \nonumber \\
\text{s.t.}\quad
& \K^+\in\mathcal K_{\mathrm{dist}}^{\ell_2}(\cG^+,\F^+),
\label{eq:prob_pnp_update_stab}\\
& \K^{+[i]}\equiv\K^{[i]},
\quad
\forall i\in(\cV\cap\cV^+)\setminus\cA(\cG,\cG^+).
\label{eq:prob_pnp_update_local}
\end{align}
\end{subequations}
If further performance improvement is desired, a subset of the affected
controllers can subsequently be re-optimized while the remaining local controllers are kept fixed. Both problems involve stability constraints on infinite-dimensional
closed-loop maps, which standard neural-network parametrizations do not
enforce: unconstrained parameter updates may therefore destabilize the
network.

\section{Distributed Performance Boosting}\label{sec:proposed}

To address \eqref{eq:prob_design}--\eqref{eq:prob_pnp_update}, we
realize the IMC operator as an interconnection of local neural operators
and develop free parametrizations enforcing a prescribed network gain,
together with their PnP recertification.

\subsection{Nominal IMC reformulation and robust extension}
\label{subsec:imc_reformulation}

Consider the plant equation \eqref{eq:plant-operator}. In the IMC architecture~\cite{furieri2024learning}, the controller reconstructs the residual from the measured state history and past inputs as
\begin{equation}\label{eq:innovation_reconstructed}
    \hat{\evec}=\x-\F(\x,\uvec).
\end{equation}
We then introduce a causal auxiliary operator $\mathbf R:\ell^n\to\ell^m$ and define the control action as
\begin{equation}\label{eq:imc_law}
    \uvec=\mathbf R(\hat{\evec}).
\end{equation}
Combining \eqref{eq:innovation_reconstructed} and \eqref{eq:imc_law} gives the implicit feedback representation $\uvec=\mathbf R(\x-\F(\x,\uvec))$, which induces a causal policy $\uvec=\K_{\mathbf R}(\x,\uvec)$.

In the centralized setting, under the exact-model assumption,
$\hat{\evec}=\evec$, if
$\mathbf R\in\mathcal L_2$, the closed loop
\eqref{eq:plant-operator}--\eqref{eq:imc_law} is
$\ell_2$-stable. Conversely, any stabilizing centralized policy can be
represented by a suitable operator $\mathbf R\in\mathcal L_2$~\cite{furieri2024learning}. 

The same characterization applies to
$\cG$-distributed policies under the distributed implementation notion introduced
in Section~\ref{sec:prelim}. Indeed, in the centralized converse, $\mathbf R$ is
chosen as the closed-loop map from $\evec$ to $\uvec$ induced by the
stabilizing policy. If this policy is $\cG$-distributed, this map is
also $\cG$-distributed, since it can be implemented by the local plant
dynamics interconnected with the distributed controller. This observation concerns distributed implementability rather than
static sparsity of the input--output map; imposing the latter directly
on a Youla-type parameter generally requires additional invariance
properties~\cite{rotkowitz2005characterization,furieri2020sparsity}.

However, this stability guarantee is strictly \emph{nominal}, meaning
that it assumes the model used inside the controller exactly matches the
true pre-stabilized plant operator $\F$. Under this assumption,
$\hat{\evec}=\evec$. In practice, the controller relies on a nominal
model $\widehat\F$, which may differ from the true plant. We model this
mismatch as
\begin{equation}\label{eq:model_mismatch}
    \x
    =
    \widehat\F(\x,\uvec)
    +
    \boldsymbol{\Delta}(\x,\uvec)
    +
    \mathbf E(x_0)+\mathbf S\w ,
\end{equation}
where
$\boldsymbol{\Delta}:\ell^n\times\ell^m\to\ell^n$
is an $\mathcal L_2$ strictly causal operator capturing the unmodeled
dynamics. Building on the small-gain arguments of
\cite{furieri2024learning}, the following proposition provides a robust
IMC condition that explicitly accounts for nonzero initial conditions.

\begin{proposition}
\label{prop:robust_imc}
Suppose there exist constants
$\gamma_{\F},\gamma_{\boldsymbol{\Delta}}>0$ and a function
$\beta_F\in\mathcal K$ such that, for every
$x_0\in\R^n$ and all $\uvec,\w\in\ell_2$,
the true plant map and the mismatch satisfy
\begin{align}
    \|\x\|
    &\le
    \gamma_{\F}
    \bigl(\|\uvec\|+\|\w\|\bigr)
    +
    \beta_F(\|x_0\|),
    \label{eq:true_plant_gain}\\
    \|\boldsymbol{\Delta}(\x,\uvec)\|
    &\le
    \gamma_{\boldsymbol{\Delta}}
    \bigl(\|\x\|+\|\uvec\|\bigr).
    \label{eq:mismatch_gain}
\end{align}
Let $\mathbf R\in\mathcal L_2$ satisfy
\[
\|\mathbf R(\mathbf r)\|
\le
\gamma_{\mathbf R}\|\mathbf r\|+\beta_R
\]
for some $\gamma_{\mathbf R}>0$ and $\beta_R\ge0$. If
\begin{equation}\label{eq:robust_imc_condition}
\gamma_{\boldsymbol{\Delta}}
(\gamma_{\F}+1)\gamma_{\mathbf R}<1,
\end{equation}
then the closed loop defined by
$\uvec=\mathbf R(\x-\widehat\F(\x,\uvec))$
is finite-gain $\ell_2$-stable, with an additive bound that may depend
on $x_0$ and $\beta_R$.
\end{proposition}

The proof is reported in Appendix~\ref{app:proof_robust_imc}.
Proposition~\ref{prop:robust_imc} shows that robustness to model
mismatch depends explicitly on the gain of the auxiliary operator
$\mathbf R$. We therefore prescribe a target gain
$\gamma_{\mathbf R}>0$ and define
\begin{equation*}
\begin{aligned}
\mathcal R_{\gamma_{\mathbf R}}^{\ell_2}(\cG)
:=
\Big\{\mathbf R\in\mathcal L_2\ & \Big|\;
\mathbf R\text{ is $\cG$-distributed},\ \exists\,\beta_R\ge0:\\
&\|\mathbf R(\mathbf r)\|
\le\gamma_{\mathbf R}\|\mathbf r\|+\beta_R,
\  \forall\,\mathbf r\in\ell_2^n
\Big\}.
\end{aligned}
\end{equation*}
The choice of $\gamma_{\mathbf R}$ trades robustness against controller
flexibility: smaller values increase the margin in
\eqref{eq:robust_imc_condition}, but restrict the admissible
performance-boosting operators. At this stage, $\beta_R$ is simply an
additive term in the input--output bound. For the dynamic state-space
realizations introduced below, it accounts for a possibly nonzero
controller initial state without affecting
\eqref{eq:robust_imc_condition}.

For a prescribed $\gamma_{\mathbf R}$ satisfying~\eqref{eq:robust_imc_condition}, the design problem is reduced to searching over
$\mathcal R_{\gamma_{\mathbf R}}^{\ell_2}(\cG)$. This characterization is
still infinite-dimensional, and the prescribed gain is a constraint on
the complete interconnected operator. Directly parametrizing a neural
controller would therefore require enforcing this constraint during
training, for instance through constrained optimization or projection. \textit{Free parametrizations} provide an optimization-oriented alternative:
unconstrained finite-dimensional variables are mapped to dynamic
operators satisfying the required gain condition for every parameter
value~\cite{furieri2024learning,manchester2026neural,massai2025free}.
In the distributed setting, however, freely parametrizing each local
operator is not sufficient by itself, since their interconnection must
also satisfy the prescribed gain of the overall operator $\mathbf R$.

\begin{figure}
    \centering
\tikzset{every picture/.style={line width=0.75pt}} 
\begin{tikzpicture}[x=0.75pt,y=0.75pt,yscale=-0.8,xscale=0.83, thick,scale=0.85, every node/.style={scale=0.85}]

\draw   (44,95) -- (153,95) -- (153,125) -- (44,125) -- cycle ;
\draw    (152.67,110.17) -- (185,110.01) ;
\draw [shift={(188,110)}, rotate = 179.73] [fill={rgb, 255:red, 0; green, 0; blue, 0 }  ][line width=0.08]  [draw opacity=0] (10.72,-5.15) -- (0,0) -- (10.72,5.15) -- (7.12,0) -- cycle    ;
\draw   (188,110) .. controls (188,104.48) and (192.48,100) .. (198,100) .. controls (203.52,100) and (208,104.48) .. (208,110) .. controls (208,115.52) and (203.52,120) .. (198,120) .. controls (192.48,120) and (188,115.52) .. (188,110) -- cycle ;
\draw    (198.17,62.5) -- (198.01,97) ;
\draw [shift={(198,100)}, rotate = 270.25] [fill={rgb, 255:red, 0; green, 0; blue, 0 }  ][line width=0.08]  [draw opacity=0] (10.72,-5.15) -- (0,0) -- (10.72,5.15) -- (7.12,0) -- cycle    ;
\draw   (246,165) -- (293,165) -- (293,189.5) -- (246,189.5) -- cycle ;
\draw    (320,177) -- (297,177) ;
\draw [shift={(294,177)}, rotate = 360] [fill={rgb, 255:red, 0; green, 0; blue, 0 }  ][line width=0.08]  [draw opacity=0] (10.72,-5.15) -- (0,0) -- (10.72,5.15) -- (7.12,0) -- cycle    ;
\draw   (320,177) .. controls (320,171.48) and (324.48,167) .. (330,167) .. controls (335.52,167) and (340,171.48) .. (340,177) .. controls (340,182.52) and (335.52,187) .. (330,187) .. controls (324.48,187) and (320,182.52) .. (320,177) -- cycle ;
\draw    (360.5,177) -- (343,177) ;
\draw [shift={(340,177)}, rotate = 360] [fill={rgb, 255:red, 0; green, 0; blue, 0 }  ][line width=0.08]  [draw opacity=0] (10.72,-5.15) -- (0,0) -- (10.72,5.15) -- (7.12,0) -- cycle    ;
\draw   (360.5,162) -- (466.75,162) -- (466.75,192) -- (360.5,192) -- cycle ;
\draw    (29.67,100.5) -- (41,100.5) ;
\draw [shift={(44,100.5)}, rotate = 180] [fill={rgb, 255:red, 0; green, 0; blue, 0 }  ][line width=0.08]  [draw opacity=0] (10.72,-5.15) -- (0,0) -- (10.72,5.15) -- (7.12,0) -- cycle    ;
\draw    (29,120) -- (42,120) ;
\draw [shift={(45,120)}, rotate = 180] [fill={rgb, 255:red, 0; green, 0; blue, 0 }  ][line width=0.08]  [draw opacity=0] (10.72,-5.15) -- (0,0) -- (10.72,5.15) -- (7.12,0) -- cycle    ;
\draw    (330,110.33) -- (330,164) ;
\draw [shift={(330,167)}, rotate = 270] [fill={rgb, 255:red, 0; green, 0; blue, 0 }  ][line width=0.08]  [draw opacity=0] (10.72,-5.15) -- (0,0) -- (10.72,5.15) -- (7.12,0) -- cycle    ;
\draw    (483,165.5) -- (470,165.91) ;
\draw [shift={(467,166)}, rotate = 358.21] [fill={rgb, 255:red, 0; green, 0; blue, 0 }  ][line width=0.08]  [draw opacity=0] (10.72,-5.15) -- (0,0) -- (10.72,5.15) -- (7.12,0) -- cycle    ;
\draw    (483,54) -- (483,165.5) ;
\draw    (30.17,176.17) -- (245.67,177) ;
\draw    (30.13,54.25) -- (30.13,100.75) ;
\draw    (30.13,54.25) -- (483,54) ;
\draw    (30.17,119.67) -- (30.17,176.17) ;
\draw    (208,110) -- (354,110.49) ;
\draw [shift={(357,110.5)}, rotate = 180.19] [fill={rgb, 255:red, 0; green, 0; blue, 0 }  ][line width=0.08]  [draw opacity=0] (10.72,-5.15) -- (0,0) -- (10.72,5.15) -- (7.12,0) -- cycle    ;
\draw    (317.33,76.17) -- (353.33,76.17) ;
\draw [shift={(356.33,76.17)}, rotate = 180] [fill={rgb, 255:red, 0; green, 0; blue, 0 }  ][line width=0.08]  [draw opacity=0] (10.72,-5.15) -- (0,0) -- (10.72,5.15) -- (7.12,0) -- cycle    ;
\draw    (360.67,90.17) -- (482.5,90) ;
\draw  [color={rgb, 255:red, 0; green, 59; blue, 255 }  ,draw opacity=1 ][dash pattern={on 3.75pt off 3pt on 7.5pt off 1.5pt}] (221,140.67) -- (489,140.67) -- (489,216.17) -- (221,216.17) -- cycle ;
\draw  [color={rgb, 255:red, 0; green, 59; blue, 255 }  ,draw opacity=1 ][dash pattern={on 3.75pt off 3pt on 7.5pt off 1.5pt}] (24.75,33.67) -- (217.33,33.67) -- (217.33,134.67) -- (24.75,134.67) -- cycle ;
\draw    (483,185.67) -- (470.33,185.8) ;
\draw [shift={(467.33,185.83)}, rotate = 359.39] [fill={rgb, 255:red, 0; green, 0; blue, 0 }  ][line width=0.08]  [draw opacity=0] (10.72,-5.15) -- (0,0) -- (10.72,5.15) -- (7.12,0) -- cycle    ;
\draw    (233,207.5) -- (483,208.33) ;
\draw    (233,207.5) -- (233,177.5) ;
\draw    (483,208.33) -- (483,185.67) ;
\draw  [fill={rgb, 255:red, 0; green, 0; blue, 0 }  ,fill opacity=1 ] (357.9,63.68) -- (360,63.67) -- (360.36,115.83) -- (358.26,115.84) -- cycle ;
\draw    (269.33,124.33) -- (269.02,161.5) ;
\draw [shift={(269,164.5)}, rotate = 270.48] [fill={rgb, 255:red, 0; green, 0; blue, 0 }  ][line width=0.08]  [draw opacity=0] (10.72,-5.15) -- (0,0) -- (10.72,5.15) -- (7.12,0) -- cycle    ;
\draw  [color={rgb, 255:red, 208; green, 2; blue, 27 }  ,draw opacity=1 ][dash pattern={on 3.75pt off 3pt on 7.5pt off 1.5pt}] (221,33.67) -- (489,33.67) -- (489,134.67) -- (221,134.67) -- cycle ;

\draw (44.5,99.23) node [anchor=north west][inner sep=0.75pt]  [font=\normalsize]  {$\mathbf{F}^{[i]}\left( \mathbf{x}^{[\mathcal{N}_{i}]} ,\mathbf{u}^{[i]}\right)$};
\draw (245.5,166) node [anchor=north west][inner sep=0.75pt]    {$\Erre^{[ i]}( \cdot )$};
\draw (164.67,59.9) node [anchor=north west][inner sep=0.75pt]    {$\mathbf{e}^{[i]}$};
\draw (295,150) node [anchor=north west][inner sep=0.75pt]    {$\widehat{\mathbf{e}}^{[i]}$};
\draw (240.83,87.9) node [anchor=north west][inner sep=0.75pt]    {$\mathbf{x}^{[ i]}$};
\draw (171.5,87.07) node [anchor=north west][inner sep=0.75pt]    {$+$};
\draw (202.17,76.4) node [anchor=north west][inner sep=0.75pt]    {$+$};
\draw (298.33,57.07) node [anchor=north west][inner sep=0.75pt]    {$\mathbf{x}^{[\mathcal{N}_{i} \backslash i]}$};
\draw (333.5,142.4) node [anchor=north west][inner sep=0.75pt]    {$+$};
\draw (344,181) node [anchor=north west][inner sep=0.75pt]    {$-$};
\draw (95.33,36.33) node [anchor=north west][inner sep=0.75pt]   [align=left] {\textcolor[rgb]{0.29,0.56,0.89}{{\fontfamily{ptm}\selectfont {\small Subsystem}}}};
\draw (363,135) node [anchor=north west][inner sep=0.75pt]   [align=left] {\begin{minipage}[lt]{61.6pt}\setlength\topsep{0pt}
\begin{center}
\textcolor[rgb]{0.29,0.56,0.89}{{\fontfamily{ptm}\selectfont {\small Local controller}}}
\end{center}

\end{minipage}};
\draw (194,157.4) node [anchor=north west][inner sep=0.75pt]    {$\mathbf{u}^{[i]}$};
\draw (235.5,35.33) node [anchor=north west][inner sep=0.75pt]   [align=left] {\begin{minipage}[lt]{151.68pt}\setlength\topsep{0pt}
\begin{center}
{\fontfamily{ptm}\selectfont {\small \textcolor[rgb]{0.82,0.01,0.11}{Neighbor-to-neighbor communication}}}
\end{center}
\end{minipage}};
\draw (360,165.65) node [anchor=north west][inner sep=0.75pt]  [font=\normalsize]  {$\mathbf{F}^{[i]}\left(\mathbf{x}^{[\mathcal{N}_{i}]},\mathbf{u}^{[i]}\right)$};
\draw (378.33,67.07) node [anchor=north west][inner sep=0.75pt]    {$\mathbf{x}^{[\mathcal{N}_{i}]}$};
\draw (272.33,115) node [anchor=north west][inner sep=0.75pt]    {$\mathbf{z}^{[\mathcal{N}_{i} \backslash i]}$};
\end{tikzpicture}
    \caption{Local IMC implementation at subsystem $i$. The residual
    $\hat{\evec}^{[i]}$ is reconstructed from the local model, while
    internal controller outputs from one-hop neighbors provide dynamic
    coordination. The complete neural interconnection realizes
    $\Erre\in\mathcal R_{\gamma_{\mathbf R}}^{\ell_2}(\cG)$.}
    \label{fig:localContr}
\end{figure}
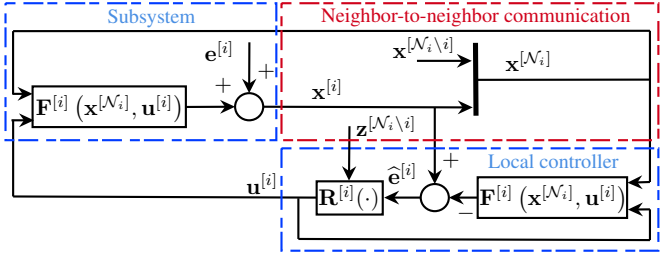

\subsection{Distributed internal realization of $\mathbf R$}
\label{subsec:cert_param}

Imposing a distributed structure directly on the input--output map of
$\mathbf R$ can be restrictive. For instance, a block-diagonal
$\mathbf R$, as considered in~\cite{furieri2024learning}, is directly
implementable but prevents dynamic coordination between subsystems.
We instead realize $\mathbf R$ as an interconnection of local neural
operators $\mathbf R^{[i]}$, one associated with each subsystem, and
impose the communication structure on this internal realization.
Each local operator exchanges internal signals only with neighboring
subsystems, allowing information to propagate
through the network over time while preserving a neighbor-to-neighbor
implementation.

To construct this (see Fig.~\ref{fig:interconnectionREN}), let the external input $\hat{\evec}$, internal input $\mathbf{v}$, internal output $\mathbf{z}$, and control $\uvec$ have pointwise dimensions $n,q,r,$ and $m$, respectively, with $q=\sum_{i\in\cV}q_i$ and $r=\sum_{i\in\cV}r_i$. We consider a local operator
$\mathbf R^{[i]}:\ell^{q_i}\to\ell^{r_i}$ with state-space
realization
\begin{equation}\label{eq:nonlinsuboperlocal}
\begin{aligned}
\xi^{[i]}_{t+1}
&=
\omega^{[i]}(\xi^{[i]}_t,v^{[i]}_t;\theta^{[i]}),\\
z^{[i]}_t
&=
\chi^{[i]}(\xi^{[i]}_t,v^{[i]}_t;\theta^{[i]}),
\qquad t\ge0,
\end{aligned}
\end{equation}
where $\xi_t^{[i]}\in\R^{n_{\xi_i}}$ is the internal controller state
and $\theta^{[i]}\in\R^{p_i}$ collects the trainable parameters of the
local operator. The mappings $\omega^{[i]}$ and $\chi^{[i]}$ may
include neural networks parametrized by $\theta^{[i]}$; their
architectures and parameter dimensions $p_i$ may differ across
subsystems. 
We consider the following assumption on these local operators. 

\begin{assumption}\label{ass:local_diss}\textup{\textbf{(Free parametrization with gain bound)}}
For each $i\in\cV$ and any prescribed gain $\gamma^{[i]}>0$, the
parametrization of $\mathbf R^{[i]}$ is such that, for every
$\theta^{[i]}\in\R^{p_i}$, there exists a nonnegative storage function
$V^{[i]}$ with $V^{[i]}(0)=0$ satisfying $\forall t\ge0$
\begin{equation}\label{eq:local_dissipation}
V^{[i]}(\xi^{[i]}_{t+1})-V^{[i]}(\xi^{[i]}_t)
\le
(\gamma^{[i]})^2\|v^{[i]}_t\|^2-\|z^{[i]}_t\|^2.
\end{equation}
\end{assumption}
Here $\gamma^{[i]}$ bounds the gain of the local operator
$\mathbf R^{[i]}$ from $\mathbf v^{[i]}$ to $\mathbf z^{[i]}$,
whereas $\gamma_{\mathbf R}$ denotes the gain of the complete
interconnection from $\hat{\evec}$ to $\uvec$. For each prescribed $\gamma^{[i]}$, the realization
\eqref{eq:nonlinsuboperlocal} is obtained from the corresponding free
parametrization; we leave the dependence of $\omega^{[i]}$ and
$\chi^{[i]}$ on $\gamma^{[i]}$ implicit.
Assumption~\ref{ass:local_diss} means that, once a local gain
$\gamma^{[i]}$ is prescribed, the parameters $\theta^{[i]}$ can be
trained freely without violating that gain bound. This property is
provided by several free parametrizations of centralized stable neural
operators that admit state-space realizations of the form
\eqref{eq:nonlinsuboperlocal}, including Recurrent Equilibrium Networks
(RENs)~\cite{revay2023recurrent}, stable state-space
models~\cite{massai2025free,massai2025l2ru}, and dissipative
port-Hamiltonian architectures~\cite{zakwan2026neural}.

Our first step toward a distributed free parametrization is to use
these centralized constructions as certified local building blocks and
compose their gain guarantees across the network. To this end, we
associate each local operator with a nonnegative multiplier $\alpha^{[i]}\geq0$,
which weights the local inequality~\eqref{eq:local_dissipation}.
The pair $(\alpha^{[i]},\gamma^{[i]})$ will be selected below so that
the interconnected operator satisfies a desired global gain bound.

Stacking the local variables yields
$\mathbf z=\mathbf R_{\rm loc}(\mathbf v)$, where
$\mathbf R_{\rm loc}
:=
\blkdiag(\mathbf R^{[1]},\ldots,\mathbf R^{[N]})$. The global internal
interconnection is defined by
\begin{equation}\label{eq:interconn_R}
\mathbf v=M_{vz}\mathbf z+M_{ve}\hat{\evec},
\quad
\uvec=M_{uz}\mathbf z,
\quad
\mathbf z=\mathbf R_{\rm loc}(\mathbf v).
\end{equation}
The matrices $M_{ve}\in\mathbb{R}^{q\times n}$, $M_{vz}\in\mathbb{R}^{q\times r}$, and $M_{uz}\in\mathbb{R}^{m\times r}$ determine, respectively,
how reconstructed residuals are injected into the local operators, how their
outputs are exchanged, and how they are mapped to the control input.
Throughout, we require these matrices to respect the communication
graph: with a subsystem-wise partition,
\begin{equation}\label{eq:graph_compatibility}
M^{[ij]}\neq0
\ \Longrightarrow\
j\in\cN_i,
\quad
M\in\{M_{ve},M_{vz},M_{uz}\}.
\end{equation}
We call interconnection matrices satisfying
\eqref{eq:graph_compatibility} \emph{graph-compatible}.
This guarantees that the realization of $\mathbf R$ is
$\cG$-distributed. We assume throughout that the interconnection~\eqref{eq:interconn_R}
is well posed: for every time $t$ and admissible
$(\xi_t,\hat e_t)$, the algebraic equations induced by the
interconnection admit a unique solution for $(v_t,z_t,u_t)$.
This condition is automatically satisfied when the local operators are
strictly causal, since $z_t$ is independent of $v_t$. 

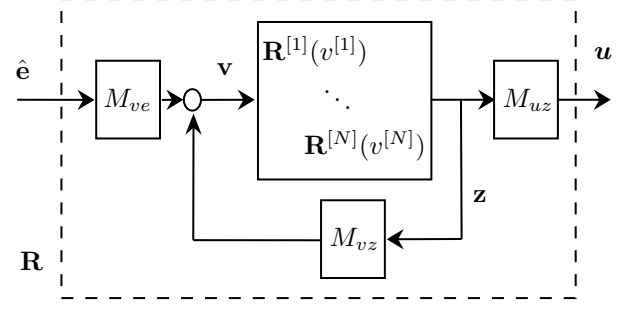
\begin{figure}
    \centering
\tikzset{every picture/.style={line width=0.75pt}} 
\begin{tikzpicture}[x=0.75pt,y=0.75pt,yscale=-1,xscale=0.8]
\draw   (241,41) -- (349.75,41) -- (349.75,120) -- (241,120) -- cycle ;
\draw   (280.57,130.43) -- (320.57,130.43) -- (320.57,169.43) -- (280.57,169.43) -- cycle ;
\draw    (349.4,80) -- (386.75,80) ;
\draw [shift={(389.75,80)}, rotate = 180] [fill={rgb, 255:red, 0; green, 0; blue, 0 }  ][line width=0.08]  [draw opacity=0] (10.72,-5.15) -- (0,0) -- (10.72,5.15) -- (7.12,0) -- cycle    ;
\draw    (368.71,149.86) -- (368.31,80.4) ;
\draw    (368.71,149.86) -- (324.75,149.99) ;
\draw [shift={(321.75,150)}, rotate = 359.83] [fill={rgb, 255:red, 0; green, 0; blue, 0 }  ][line width=0.08]  [draw opacity=0] (10.72,-5.15) -- (0,0) -- (10.72,5.15) -- (7.12,0) -- cycle    ;
\draw    (205.24,80.08) -- (236,80.01) ;
\draw [shift={(239,80)}, rotate = 179.87] [fill={rgb, 255:red, 0; green, 0; blue, 0 }  ][line width=0.08]  [draw opacity=0] (10.72,-5.15) -- (0,0) -- (10.72,5.15) -- (7.12,0) -- cycle    ;
\draw   (194.51,80.08) .. controls (194.51,77.12) and (196.92,74.71) .. (199.88,74.71) .. controls (202.84,74.71) and (205.24,77.12) .. (205.24,80.08) .. controls (205.24,83.04) and (202.84,85.44) .. (199.88,85.44) .. controls (196.92,85.44) and (194.51,83.04) .. (194.51,80.08) -- cycle ;
\draw    (200.43,150.43) -- (200.58,89.87) ;
\draw [shift={(200.59,86.87)}, rotate = 90.14] [fill={rgb, 255:red, 0; green, 0; blue, 0 }  ][line width=0.08]  [draw opacity=0] (10.72,-5.15) -- (0,0) -- (10.72,5.15) -- (7.12,0) -- cycle    ;
\draw    (280.75,150.5) -- (200.43,150.43) ;
\draw    (90,80) -- (135.51,80.07) ;
\draw [shift={(138.51,80.08)}, rotate = 180.09] [fill={rgb, 255:red, 0; green, 0; blue, 0 }  ][line width=0.08]  [draw opacity=0] (10.72,-5.15) -- (0,0) -- (10.72,5.15) -- (7.12,0) -- cycle    ;
\draw  [dash pattern={on 4.5pt off 4.5pt}] (117.75,29.5) -- (440.25,29.5) -- (440.25,179.5) -- (117.75,179.5) -- cycle ;
\draw   (389.07,60.43) -- (429.07,60.43) -- (429.07,99.43) -- (389.07,99.43) -- cycle ;
\draw    (429.4,80) -- (459.25,80) ;
\draw [shift={(462.25,80)}, rotate = 180] [fill={rgb, 255:red, 0; green, 0; blue, 0 }  ][line width=0.08]  [draw opacity=0] (10.72,-5.15) -- (0,0) -- (10.72,5.15) -- (7.12,0) -- cycle    ;
\draw   (139.57,60.43) -- (179.57,60.43) -- (179.57,99.43) -- (139.57,99.43) -- cycle ;
\draw    (180.74,80.08) -- (191.25,80.02) ;
\draw [shift={(194.25,80)}, rotate = 179.67] [fill={rgb, 255:red, 0; green, 0; blue, 0 }  ][line width=0.08]  [draw opacity=0] (10.72,-5.15) -- (0,0) -- (10.72,5.15) -- (7.12,0) -- cycle    ;
\draw (242,46.6) node [anchor=north west][inner sep=0.75pt]    {${\textstyle \Erre^{[1]}}(v^{[1]})$};
\draw (268,93.9) node [anchor=north west][inner sep=0.75pt]    {$\Erre^{[N]}(v^{[N]})$};
\draw (279.2,64.4) node [anchor=north west][inner sep=0.75pt]    {$\ddots $};
\draw (284.27,142.83) node [anchor=north west][inner sep=0.75pt]    {$M_{vz}$};
\draw (374,123.9) node [anchor=north west][inner sep=0.75pt]    {$\mathbf{z}$};
\draw (213.29,60.17) node [anchor=north west][inner sep=0.75pt]    {$\mathbf{v}$};
\draw (449.9,52) node [anchor=north west][inner sep=0.75pt]    {$\boldsymbol{u}$};
\draw (87.3,58.57) node [anchor=north west][inner sep=0.75pt]    {$\hat{\evec}$};
\draw (392.77,72.83) node [anchor=north west][inner sep=0.75pt]    {$M_{uz}$};
\draw (90,154.9) node [anchor=north west][inner sep=0.75pt]    {$\Erre$};
\draw (142.27,72.83) node [anchor=north west][inner sep=0.75pt]    {$M_{ve}$};
\end{tikzpicture}
    \caption{Internal realization of $\mathbf R$. The matrix
    $M_{ve}$ injects the external residual $\hat{\evec}$, $M_{vz}$
    routes neighboring local-operator outputs $\mathbf z$ back to local-operator inputs
    $\mathbf v$, and $M_{uz}$ reads the control input $\uvec$ from the local-operator outputs.}
    \label{fig:interconnectionREN}
\end{figure}

We next combine the local gain inequalities into a sufficient
condition for the gain from $\hat{\evec}$ to $\uvec$. For a gain
$\gamma_{\mathbf R}>0$ and nonnegative multipliers
$\alpha^{[i]}\ge0$, define
\begin{align}
\Pi_{N,v}
&:=
\blkdiag\{
\alpha^{[1]}(\gamma^{[1]})^2 I_{q_1},
\ldots,
\alpha^{[N]}(\gamma^{[N]})^2 I_{q_N}
\},
\label{eq:Pi_def}\\
\Pi_{N,z}
&:=
\blkdiag\{
\alpha^{[1]} I_{r_1},
\ldots,
\alpha^{[N]} I_{r_N}
\}.
\label{eq:Pi_def2}
\end{align}
Multiplying \eqref{eq:local_dissipation} by
$\alpha^{[i]}$ and summing over $i$ yields
\begin{align}
    V_\alpha(\xi_{t+1})-V_\alpha(\xi_t)
&\le
v_t^\top\Pi_{N,v}v_t-z_t^\top\Pi_{N,z}z_t, \label{eq:overall_storage_ineq}
\\
V_\alpha(\xi_t)&:=
\sum_{i\in\mathcal V}\alpha^{[i]}V^{[i]}(\xi_t^{[i]}), \label{eq:overall_storage}
\end{align}
where $\xi_t:=\col_{i\in\cV}(\xi_t^{[i]})$. Thus, $\Pi_{N,v}$ and $\Pi_{N,z}$ collect the weighted local gain
inequalities. The following result gives a sufficient condition for
their interconnection to have the prescribed gain
$\gamma_{\mathbf R}$ from $\hat{\evec}$ to $\uvec$.

\begin{proposition} \label{prop:stabL2}
Let Assumption~\ref{ass:local_diss} hold and assume that the
interconnection~\eqref{eq:interconn_R} is well-posed. If there exist
$\alpha^{[i]}\ge0$ such that
\begin{equation}\label{eq:conditionL22}
G^\top QG\preceq0,
\end{equation}
where
\begin{equation*}
G:=
\begin{bmatrix}
M_{vz} & M_{ve}\\
I & 0\\
0 & I\\
M_{uz} & 0
\end{bmatrix},
\quad
Q:=
\begin{bmatrix}
\Pi_{N,v} & 0 & 0 & 0\\
0 & -\Pi_{N,z} & 0 & 0\\
0 & 0 & -\gamma_{\mathbf R}^2 I_n & 0\\
0 & 0 & 0 & I_m
\end{bmatrix},
\end{equation*}
then the operator $\mathbf R$ induced by
\eqref{eq:interconn_R} satisfies
\begin{equation}\label{eq:R_gain_bound}
\|\uvec\|
\le
\gamma_{\mathbf R}\|\hat{\evec}\|
+\beta_R,
\end{equation}
with $\beta_R:=\sqrt{V_\alpha(\xi_0)}$.
In particular, $\beta_R=0$ when $\xi_0=0$.
\end{proposition}
The result follows from the compositional dissipativity
argument in~\cite[Prop.~5.1]{arcak2016networks}. 
Summing
\eqref{eq:overall_storage_ineq} under \eqref{eq:conditionL22} gives
\eqref{eq:R_gain_bound}. 
We call \eqref{eq:conditionL22} the \emph{network gain certificate}, since it guarantees the prescribed gain of the interconnected operator $\mathbf R$ required by Proposition~\ref{prop:robust_imc} for robust closed-loop stability. Since $\theta^{[i]}$ can be trained freely for fixed $\gamma^{[i]}$,
it remains to select $(\alpha^{[i]},\gamma^{[i]})$ so that this
condition holds. 

\subsection{Free parametrization of the network gain certificate}
\label{sec:param}

We now construct explicit maps from unconstrained variables to the local
gains $\gamma^{[i]}$ and multipliers $\alpha^{[i]}$ so that the network
gain certificate~\eqref{eq:conditionL22} is satisfied by design.
The first construction uses only coupling information associated with
each local operator and the neighboring subsystems. The second uses an upper bound on the
network norm $\|M_{vz}\|_2$, which can reduce conservatism but requires
network-wide information. In both cases, all variables exposed to the
optimizer are unconstrained. The additional structure below on
$M_{ve}$ and $M_{uz}$ is used only to obtain these explicit maps.
\begin{assumption}\label{ass:MzMe}
The matrices $M_{ve}$ and $M_{uz}$ satisfy
\begin{enumerate}
    \item there exists a permutation matrix $C$ such that
    \begin{equation}\label{eq:Mve_structure_again}
    C M_{ve}
    =
    \begin{bmatrix}
    I_n\\
    0
    \end{bmatrix},
    \qquad
    M_{ve}^\top M_{ve}=I_n;
    \end{equation}
    \item the columns of $M_{uz}$ are orthogonal, namely
    \begin{equation}\label{eq:Muz_structure_again}
    M_{uz}^\top M_{uz}
    =
    H
    =
    \diag(h),
    \qquad h\ge0.
    \end{equation}
\end{enumerate}
\end{assumption}
Assumption~\ref{ass:MzMe} restricts only how signals enter and leave
the interconnected local neural operators. Condition~(i) means that every
component of the reconstructed residual is injected, with unit gain,
into exactly one internal input channel; additional internal channels
may receive no residual directly. Condition~(ii) means that different components of $\mathbf z$ contribute to mutually orthogonal
directions of the control input. For example, $M_{ve}=I$ and $M_{uz}=I$ simply feed each residual component to one local-operator input and use each
local-operator output
as one control component. Together with graph compatibility, these conditions leave $M_{vz}$ as the matrix describing the coupling between local operators.

To express the network gain condition in terms of subsystem-wise quantities, we first introduce notation that distinguishes the local input channels receiving the reconstructed residual and summarizes how the local outputs enter the control input.
Let $m_{jk}$ denote the $(j,k)$ entry of $M_{vz}$. For each subsystem $i$, let $\mathcal J_{v_i}$ and $\mathcal J_{z_i}$ denote the sets of coordinate indices of the pointwise stacked vectors $v_t$ and $z_t$ corresponding to $v_t^{[i]}$ and $z_t^{[i]}$, respectively.
We define
\begin{equation}\label{eq:Jvi1_def}
\mathcal J_{v_i}^1
:=
\left\{
j\in\mathcal J_{v_i}
\;:\;
\text{row }j\text{ of }M_{ve}\text{ is nonzero}
\right\},
\end{equation}
and
$\mathcal J_{v_i}^0
:=
\mathcal J_{v_i}\setminus\mathcal J_{v_i}^1$.
Thus, $\mathcal J_{v_i}^1$ contains the input coordinates of $\mathbf R^{[i]}$ that receive a component of the reconstructed residual through $M_{ve}$, whereas $\mathcal J_{v_i}^0$ contains the remaining input coordinates.
Finally, recalling that $h_j$ denotes the $j$-th diagonal entry of
$H=M_{uz}^\top M_{uz}$ in~\eqref{eq:Muz_structure_again}, define
\begin{equation}\label{eq:hbar_def}
\bar h_i
:=
\max_{j\in\mathcal J_{z_i}} h_j .
\end{equation}

\subsubsection{Neighbor-based compositional parametrization}
\label{subsec:param_local}

The first parametrization replaces the network gain certificate~\eqref{eq:conditionL22} with sufficient conditions based on local bounds on the interconnection matrix $M_{vz}$. Since $v=M_{vz}z+M_{ve}\hat e$, rows of $M_{vz}$ describe how the outputs $z$ contribute to each internal input coordinate, whereas columns describe how each output coordinate contributes to the internal inputs across the network. For each subsystem $i$, define
\begin{align}
\rho_z^{[i]}
&:=
\max_{j\in\mathcal J_{z_i}}
\sum_k |m_{kj}|,  \
\rho_{v,0}^{[i]}
:=
\max_{j\in\mathcal J_{v_i}^0}
\sum_k |m_{jk}|,
\label{eq:rho_v0_z_local}\\
\rho_{v,1}^{[i]}
&:=
\max_{j\in\mathcal J_{v_i}^1}
\sum_k |m_{jk}|.
\label{eq:rho_v1_local}
\end{align}
We set $\rho_{v,0}^{[i]}=0$ when
$\mathcal J_{v_i}^0=\emptyset$ and
$\rho_{v,1}^{[i]}=0$ when
$\mathcal J_{v_i}^1=\emptyset$.
Thus, $\rho_z^{[i]}$ bounds the total outgoing coupling from the output coordinates of $\mathbf R^{[i]}$, while $\rho_{v,0}^{[i]}$ and $\rho_{v,1}^{[i]}$ bound the total incoming coupling to its input coordinates without and with direct residual injection, respectively. By graph compatibility, these quantities depend only on couplings involving subsystem $i$ and its neighbors.

From~\eqref{eq:Pi_def}, the local gain $\gamma^{[i]}$ enters the network gain certificate through the product
$c^{[i]}:=\alpha^{[i]}(\gamma^{[i]})^2$.
For the neighbor-based construction, we define the following quantities associated with the two classes of input coordinates:
\begin{equation}\label{eq:Bi_local}
\mathcal B_i
:=
\left\{
\frac{\gamma_{\mathbf R}^2}
{1+\rho_{v,1}^{[i]}\gamma_{\mathbf R}^2}
\;|\;
\mathcal J_{v_i}^1\neq\emptyset
\right\}
\cup
\left\{
\frac{1}{\rho_{v,0}^{[i]}}
\;|\;
\rho_{v,0}^{[i]}>0
\right\}.
\end{equation}
We then set
\begin{equation}\label{eq:cbar_local}
\bar c^{[i]}
:=
\begin{cases}
\min \mathcal B_i, & \mathcal B_i\neq\emptyset,\\
\bar c_0^{[i]}, & \mathcal B_i=\emptyset,
\end{cases}
\end{equation}
where $\bar c_0^{[i]}>0$ is a fixed design constant used consistently
across network configurations.
\begin{theorem}\label{thm:param_local}
Fix $\gamma_{\mathbf R}>0$ and $\varepsilon>0$, and suppose
Assumptions~\ref{ass:local_diss} and~\ref{ass:MzMe} hold.
For each subsystem $i$, choose
\[
\eta^{[i]}
:=
\col(\eta_\alpha^{[i]},\eta_c^{[i]})
\in\R^2
\]
and define
\begin{align}
\alpha^{[i]}
&:=
\bar h_i+\rho_z^{[i]}
+(\eta_\alpha^{[i]})^2+\varepsilon,
\label{eq:alpha_local_compact}\\
c^{[i]}
&:=
\frac{\bar c^{[i]}}
{1+\exp(-\eta_c^{[i]})},
\label{eq:c_local_compact}\\
\gamma^{[i]}
&:=
\sqrt{\frac{c^{[i]}}{\alpha^{[i]}}}.
\label{eq:gamma_local_compact}
\end{align}
Then, for every choice of the free variables $\eta^{[i]}$, the
interconnection~\eqref{eq:interconn_R} satisfies
\eqref{eq:conditionL22}.
\end{theorem}

The proof is reported in Appendix~\ref{app:proof_param_local}. It relies on a sufficient diagonal-dominance condition for the matrix inequality~\eqref{eq:conditionL22}, where the coupling terms are bounded through the local row and column sums of $M_{vz}$. This yields a lower bound on $\alpha^{[i]}$ and an upper bound on
$c^{[i]}=\alpha^{[i]}(\gamma^{[i]})^2$.
Accordingly, $\eta_\alpha^{[i]}$ selects the slack above the lower bound on $\alpha^{[i]}$, while $\eta_c^{[i]}$ determines $c^{[i]}$ within $(0,\bar c^{[i]})$ through~\eqref{eq:c_local_compact}. The local gain is then given by
$\gamma^{[i]}=\sqrt{c^{[i]}/\alpha^{[i]}}$.
The construction uses only interconnection information associated with subsystem $i$ and its neighbors. Its conservatism stems from replacing the full network gain certificate with these local diagonal-dominance bounds.

\noindent\emph{Illustrative example.}
Consider two scalar local operators with $M_{ve}=M_{uz}=I_2$,
$\gamma_{\mathbf R}=1$, $\varepsilon=0.02$, and
\[
M_{vz}
=
\begin{bmatrix}
0.6 & 0.4\\
0.2 & -0.3
\end{bmatrix}.
\]
Since $M_{uz}=I_2$, we have $\bar h_1=\bar h_2=1$.
The column sums of $|M_{vz}|$ give
$(\rho_z^{[1]},\rho_z^{[2]})=(0.8,0.7)$, while
the row sums give
$(\rho_{v,1}^{[1]},\rho_{v,1}^{[2]})=(1,0.5)$.
Moreover, $M_{ve}=I_2$ implies that every input channel
receives the residual directly, so
$\mathcal J_{v_i}^0=\emptyset$ and $\rho_{v,0}^{[i]}=0$.
Thus, $(\bar c^{[1]},\bar c^{[2]})=(0.5,2/3)$.
Choosing $c^{[1]}=c^{[2]}=0.25$ and
$\eta_\alpha^{[1]}=\eta_\alpha^{[2]}=0$ yields
$(\alpha^{[1]},\gamma^{[1]})=(1.82,0.371)$ and
$(\alpha^{[2]},\gamma^{[2]})=(1.72,0.381)$.
Figure~\ref{fig:feas_reg} compares the resulting sufficient region with
the original network gain certificate~\eqref{eq:conditionL22}.

\begin{figure}
    \centering
    \includegraphics[width=0.9\linewidth]{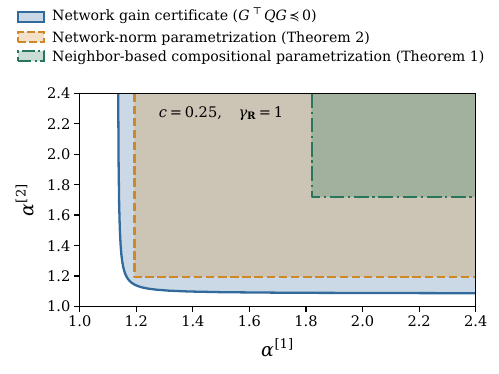}
    \caption{Feasible regions for the two-subsystem
    example with $c=\alpha^{[i]}(\gamma^{[i]})^2=0.25$:
    network gain certificate~\eqref{eq:conditionL22} (blue),
    neighbor-based parametrization (green), and network-norm
    parametrization (yellow).}
    \label{fig:feas_reg}
\end{figure}

\subsubsection{Network-norm parametrization}
\label{subsec:param_spectral}

The neighbor-based parametrization relies on separate local bounds on the couplings involving each subsystem. We now consider an alternative based on a bound on the overall interconnection matrix $M_{vz}$. By accounting for the coupling structure of the network as a whole, this can reduce conservatism, at the price of requiring network-wide information. Let
\[
\widehat\kappa\ge\|M_{vz}\|_2^2.
\]
Values of the bound $\widehat{\kappa}$ can be computed centrally or estimated through distributed
spectral algorithms~\cite{kempe2004decentralized}.

\begin{theorem}\label{thm:param_spectral}
Fix $\gamma_{\mathbf R}>0$ and $\varepsilon>0$, and suppose
Assumptions~\ref{ass:local_diss} and~\ref{ass:MzMe} hold. Let
$\widehat\kappa\ge\|M_{vz}\|_2^2$.\\
Choose $\eta_{\rm net}\in\R$ and define
\begin{equation}\label{eq:c_spectral}
c
:=
\frac{\gamma_{\mathbf R}^2}
{1+\exp(-\eta_{\rm net})},
\qquad
\tau
:=
\frac{c\gamma_{\mathbf R}^2}
{\gamma_{\mathbf R}^2-c}.
\end{equation}
For each subsystem $i$, choose $\eta_\alpha^{[i]}\in\R$ and define
\begin{align}
\alpha^{[i]}
&:=
\bar h_i+\tau \widehat\kappa
+(\eta_\alpha^{[i]})^2+\varepsilon,
\label{eq:alpha_spectral}\\
\gamma^{[i]}
&:=
\sqrt{\frac{c}{\alpha^{[i]}}}.
\label{eq:gamma_spectral}
\end{align}
Then, for every choice of $\eta_{\rm net}$ and
$\eta_\alpha^{[i]}$, the interconnection~\eqref{eq:interconn_R}
satisfies~\eqref{eq:conditionL22}.
\end{theorem}
The proof is reported in Appendix~\ref{app:proof_param_spectral}. It bounds the coupling term in the network gain certificate~\eqref{eq:conditionL22} through the network-wide quantity $\widehat\kappa$, leading to a common value
$c=\alpha^{[i]}(\gamma^{[i]})^2$ for all subsystems and a corresponding lower bound on each $\alpha^{[i]}$.
Accordingly, $\eta_{\rm net}$ selects the common $c$ within
$(0,\gamma_{\mathbf R}^2)$, while each $\eta_\alpha^{[i]}$ independently selects the slack above its local lower bound on $\alpha^{[i]}$. The local gains are then given by
$\gamma^{[i]}=\sqrt{c/\alpha^{[i]}}$.

Using the overall coupling norm instead of separate local bounds can reduce conservatism, as illustrated by the example below, at the price of requiring network-wide information. For adjacency-type interconnections, this reduction can be particularly relevant in irregular networks where a few subsystems have substantially higher degree than the others.

For the same example,
$\widehat\kappa=\|M_{vz}\|_2^2=0.52$.
With $c=0.25$ and $\eta_\alpha^{[i]}=0$,
Theorem~\ref{thm:param_spectral} gives
$\alpha^{[1]}=\alpha^{[2]}=1.193,
$ $
\gamma^{[1]}=\gamma^{[2]}=0.458$.
Thus, for this interconnection and the same $c$, the network-norm
construction admits larger local gains than the neighbor-based one, as
illustrated in Fig.~\ref{fig:feas_reg}. This ordering is not general:
the two parametrizations are different sufficient inner approximations
of~\eqref{eq:conditionL22}, and neither generally dominates the other.

\subsection{PnP recertification}
\label{sec:pnp}

We now characterize how the two free parametrizations can be updated after a PnP transition while preserving the prescribed network gain.

For the network-norm parametrization, $\widehat\kappa$ depends on the complete interconnection and must in general be recomputed when $M_{vz}$ changes. This may require updating the local gain parameters throughout the network. If, instead, a family $\mathfrak M$ of admissible interconnections is known in advance and $\bar h_i$ is unchanged for surviving subsystems, one can use a common worst-case bound $\bar\kappa
\ge
\sup_{M_{vz}\in\mathfrak M}|M_{vz}|_2^2.$
Theorem~\ref{thm:param_spectral} then certifies every configuration in the family with the same bound, so that surviving subsystems need not update their gain parameters after a PnP transition within this family. The price is the additional conservatism introduced by the worst-case network bound.

The neighbor-based parametrization instead admits localized recertification because its defining quantities depend only on couplings involving each subsystem and its neighbors. We next characterize which of these quantities must be updated after a PnP transition.

Consider a transition from $(\cG,\F)$ to $(\cG^+,\F^+)$ and denote by
$M_{vz}^+$, $M_{ve}^+$, and $M_{uz}^+$ the post-change
interconnection matrices. We assume that $M_{ve}^+$ and $M_{uz}^+$ satisfy
Assumption~\ref{ass:MzMe}, with all pre-existing blocks between surviving
subsystems unchanged; new blocks involving added subsystems may be
introduced. We further assume that the PnP transition modifies $M_{vz}$ only by adding or removing coupling blocks associated with the added or removed network connections. We denote by
$\mathbf R^+$ the resulting post-change interconnected operator. Let $\rho_z^{+[i]}$, $\rho_{v,0}^{+[i]}$, and
$\rho_{v,1}^{+[i]}$ denote the bounds
\eqref{eq:rho_v0_z_local}--\eqref{eq:rho_v1_local} evaluated for the
post-change interconnection, and let $\bar h_i^+$ and
$\bar c^{+[i]}$ denote the corresponding quantities defined in
\eqref{eq:hbar_def} and~\eqref{eq:cbar_local}. Define
\begin{multline}\label{eq:influence_region_new}
\mathcal I
:=
(\cV^+\setminus\cV)
\cup
\Big\{
i\in\cV\cap\cV^+:\\
\bar h_i^+ + \rho_z^{+[i]}
\neq
\bar h_i + \rho_z^{[i]}
\ \text{or}\
\bar c^{+[i]}\neq\bar c^{[i]}
\Big\}.
\end{multline}
Thus, $\mathcal I$ contains the new subsystems and the surviving
subsystems whose quantities entering the neighbor-based parametrization
change. Under the interconnection update described above, 
$\mathcal I\subseteq\mathcal A(\cG,\cG^+)$, since the quantities defining \(\mathcal I\) can change only for subsystems whose neighborhood changes.

We use the following recertification procedure. For every surviving
subsystem, keep the trained neural parameters unchanged, $\theta^{+[i]}=\theta^{[i]},$ $i\in\cV\cap\cV^+,$ while each newly added subsystem is assigned parameters
$\theta^{+[i]}\in\R^{p_i}$.
For every $i\in\mathcal I$, choose any
$\eta^{+[i]}\in\R^2$ and compute
$(\alpha^{+[i]},\gamma^{+[i]})$ from
\eqref{eq:alpha_local_compact}--\eqref{eq:gamma_local_compact}
using the post-change quantities. For every surviving subsystem
$i\notin\mathcal I$, retain $\eta^{+[i]}=\eta^{[i]}$, $\alpha^{+[i]}=\alpha^{[i]}$, $\gamma^{+[i]}=\gamma^{[i]}$.
For $i\in\mathcal I$, the corresponding local operator is instantiated
through its free parametrization using
$(\theta^{+[i]},\gamma^{+[i]})$; surviving local operators outside
$\mathcal I$ retain their previous realization. Thus, recertification
does not require gradient-based retraining of the surviving
$\theta^{[i]}$.

\begin{theorem}
\label{thm:pnp_locality_new}
Fix $\gamma_{\mathbf R}>0$ and $\varepsilon>0$, and suppose that
Assumptions~\ref{ass:local_diss} and~\ref{ass:MzMe} hold before and
after the PnP event. Assume that the pre-change interconnection is
certified by Theorem~\ref{thm:param_local} and apply the recertification
procedure above. Then the post-change interconnection satisfies
\eqref{eq:conditionL22} with the same target gain
$\gamma_{\mathbf R}$ 
and the resulting post-change
controller satisfies the PnP update constraints of
Definition~\ref{def:pnp_update}.
\end{theorem}

The proof is reported in Appendix~\ref{app:proof_pnp_locality}.
Thus, only subsystems in $\mathcal I$ need to recompute their certificate quantities and local gain parameters; the remaining surviving local
operators retain their previous parametrization. Subsystems in
$\mathcal A(\cG,\cG^+)$ may subsequently be retrained for performance,
but retraining is not required for recertification.
By Proposition~\ref{prop:stabL2}, the post-change operator
$\mathbf R^+$ retains the prescribed gain $\gamma_{\mathbf R}$.
Together with Assumption~\ref{ass1:prestabilization}, this gives
nominal post-change $\ell_2$ stability. Under model mismatch, the
conclusion of Proposition~\ref{prop:robust_imc} also holds if the
post-change gain bounds satisfy~\eqref{eq:robust_imc_condition}.

\section{Distributed training and PnP adaptation}
\label{sec:disttrain}

The free parametrizations of Section~\ref{sec:param} turn
\eqref{eq:prob_design} into an unconstrained finite-dimensional training
problem while preserving the prescribed network gain.

For the neighbor-based parametrization, define
\[
\vartheta^{[i]}
:=
\col(\theta^{[i]},\eta^{[i]}),
\qquad
\vartheta:=\col_{i\in\cV}\vartheta^{[i]} .
\]
For the network-norm parametrization, the local trainable variables are
$\theta^{[i]}$ and $\eta_\alpha^{[i]}$. In the distributed
implementation below, the shared variable $\eta_{\rm net}$ is fixed
before training and known to all subsystems, while
$\widehat\kappa$ is assumed available for the considered network
configuration. Optimizing $\eta_{\rm net}$ would additionally require
aggregating its gradient contributions across the network.

Given samples
$\{(x_0^{(s)},\w^{(s)})\}_{s=1}^S$
from $\mathcal D_0\times\mathcal D$, we approximate
\eqref{eq:prob_design_obj} by
\begin{equation}\label{eq:Jhat_def}
\hat{\bar J}_{T,S}(\vartheta)
=
\frac{1}{S}
\sum_{s=1}^S
J_T\!\left(
x_{0:T}^{(s)}(\vartheta),
u_{0:T-1}^{(s)}(\vartheta)
\right).
\end{equation}
Training amounts to
$\min_\vartheta\hat{\bar J}_{T,S}(\vartheta)$.
Each iteration evaluates~\eqref{eq:Jhat_def} through forward
closed-loop rollouts and differentiates them backward in time.
Since the free parametrizations enforce~\eqref{eq:conditionL22} for
every parameter value, no stability constraint or projection is
required during training.

\subsection{Neighbor-to-neighbor training}
\label{subsec:dist_bptt}

For a neighbor-to-neighbor implementation, we restrict the local
operators to be strictly causal from $v^{[i]}$ to $z^{[i]}$.
This restriction is not required by the gain and stability results of
Section~\ref{sec:proposed}.

For a single rollout, let
$J_T=\sum_{t=0}^{T-1}l(x_t,u_t)+l_T(x_T)$, with separable stage cost
\begin{equation}
l(x_t,u_t)
=
\sum_{i\in\cV}
l_i(x_t^{[\cN_i]},u_t^{[i]}),
\end{equation}
and an analogous decomposition of $l_T$ when nonzero.
Let $\zeta_t^{[i]}$ collect the local plant and controller states,
together with any locally stored signals needed to evaluate the nominal
model and reconstruct $\hat e_t^{[i]}$.
Graph compatibility and strict causality then yield the one-hop
realization
\begin{equation}\label{eq:local_rollout_map}
\zeta_{t+1}^{[i]}
=
\phi_t^{[i]}(
\zeta_t^{[\cN_i]},w_t^{[i]};
\vartheta^{[\cN_i]}).
\end{equation}
Hence the Jacobian blocks entering backpropagation through time vanish
outside the corresponding graph neighborhoods, so both forward and
reverse computations can be implemented through neighbor-to-neighbor
exchanges. Information may still propagate over multiple hops over
time.

\subsection{Training across configurations and post-PnP adaptation}

If a set of admissible PnP configurations is known in advance, training
may use rollouts generated under different interconnections. For each
rollout, the certificate quantities and local gains are evaluated for
the corresponding configuration, while the neural parameters
$\theta^{[i]}$ of surviving local operators are shared across
configurations. This allows performance to be optimized over expected
topology changes before deployment.

After an actual PnP transition,
Theorem~\ref{thm:pnp_locality_new} restores the prescribed network gain
for the neighbor-based parametrization by updating only subsystems in
$\mathcal I$, without retraining the neural parameters of surviving
subsystems. If desired, any subset
$\mathcal A_{\rm tr}\subseteq\mathcal A(\cG,\cG^+)$
can subsequently be fine-tuned for performance while the remaining
controllers retain their recertified parameters.

\begin{corollary}[Fail-safe training]
\label{cor:failsafe}
Under either Theorem~\ref{thm:param_local} or
Theorem~\ref{thm:param_spectral}, every training iterate satisfies the
prescribed network gain and therefore preserves nominal closed-loop
$\ell_2$ stability.
After a PnP transition, the same guarantee holds whenever the
post-change parameters satisfy the corresponding free parametrization;
for the neighbor-based construction this is ensured by
Theorem~\ref{thm:pnp_locality_new}.
Under model mismatch, Proposition~\ref{prop:robust_imc} applies whenever
\eqref{eq:robust_imc_condition} holds.
\end{corollary}

\emph{Proof.}
Theorems~\ref{thm:param_local} and~\ref{thm:param_spectral} guarantee
\eqref{eq:conditionL22}, and Proposition~\ref{prop:stabL2} gives the
prescribed gain of $\mathbf R$. Nominal closed-loop stability then
follows from the IMC argument, before and after PnP recertification.
The robust statement follows from Proposition~\ref{prop:robust_imc}.
\hfill$\blacksquare$

\section{Numerical example}\label{sec:num}
We consider the PnP DC microgrid of~\cite{nahata2020passivity}, equipped
with decentralized passivity-based primary controllers that preserve
stability under admissible PnP operations. A distributed neural layer is
added on top of this fixed baseline to improve transient performance.
PnP changes are handled by combining training over multiple expected
configurations with post-change recertification of the neighbor-based
parametrization.

\subsection{Microgrid setup and nominal IMC model}

\input{Figure/microgrid_setup.tex}

Each distributed generation unit (DGU) is modeled as
in~\cite{nahata2020passivity} (see Fig.~\ref{fig:DGU}) and is described
by its point of common coupling (PCC) voltage $V_i$, converter current
$I_{t,i}$, and integral state $x_{I,i}$. The augmented converter input, considered here  is
\[
V_{c,i}
=
k_{1,i}V_i+k_{2,i}I_{t,i}+k_{3,i}x_{I,i}+\Delta u_i,
\quad
\dot x_{I,i}=V_i^{\rm ref}-V_i ,
\]
where the first three terms, together with the integral-state dynamics,
define the decentralized primary controller designed to achieve
offset-free tracking of $V_i^{\rm ref}$, while $\Delta u_i$ is the
performance-boosting input.

Each DGU is connected to a ZIP load,
\begin{equation}\label{eq:zip_load}
I_{L,i}(V_i)
=
Y_{L,i}V_i+\bar I_{L,i}+\frac{P_{L,i}}{V_i},
\end{equation}
with constant-impedance, constant-current, and constant-power
components, respectively. The parameters reported in
Table~\ref{tab:microgrid_params} satisfy
$P_{L,i}<Y_{L,i}(V_i^{\rm ref})^2$, as required for local passivity
in~\cite{nahata2020passivity}. The considered load perturbations modify
only $\bar I_{L,i}$ and therefore preserve this condition. The primary
controller gains are
$k_{1,i}=0.1$, $k_{2,i}=0.5R_{t,i}$, and
$k_{3,i}=0.175(R_{t,i}/L_{t,i})$, and are kept fixed
across all load and topology configurations.

\begin{table}[t]
\centering
\caption{Microgrid parameters. The last two columns list all line
resistances appearing across the considered configurations; all line
inductances are $1$ mH.}
\label{tab:microgrid_params}
\scriptsize
\setlength{\tabcolsep}{2.3pt}
\renewcommand{\arraystretch}{0.95}

\resizebox{\columnwidth}{!}{%
\begin{tabular}{cccccccc@{\hspace{8pt}}cc}
\toprule
DGU & $R_t$ & $L_t$ & $C_t$ & $Y_L$ &
$\bar I_L$ & $P_L$ & $V^{\rm ref}$
& \multicolumn{2}{c}{Line resistances $[\Omega]$}\\
& [$\Omega$] & [mH] & [mF] & [S] & [A] & [W] & [V]
& &\\
\midrule
1 & .20 & 1.8 & 2.2 & .50 & 8.0 & 6.0 & 47.9
& $R_{12}=.050$ & $R_{16}=.100$\\
2 & .30 & 2.0 & 1.9 & .45 & 7.0 & 5.0 & 48.0
& $R_{13}=.070$ & $R_{26}=.090$\\
3 & .10 & 2.2 & 1.7 & .55 & 7.5 & 6.5 & 47.7
& $R_{24}=.040$ & $R_{36}=.110$\\
4 & .50 & 3.0 & 2.5 & .40 & 6.0 & 4.5 & 48.0
& $R_{34}=.060$ & $R_{46}=.085$\\
5 & .40 & 1.2 & 2.0 & .35 & 6.5 & 5.0 & 47.8
& $R_{45}=.080$ & $R_{56}=.080$\\
6 & .25 & 1.6 & 2.1 & .42 & 4.4 & 5.5 & 48.1
& &\\
\bottomrule
\end{tabular}}
\end{table}

Figure~\ref{fig:DCgrid} summarizes the topology changes considered in
the simulation. Initially, DGUs~1--5 are connected through the black
and red links, with edge set
$\mathcal E^-=\{(1,2),\allowbreak(1,3),\allowbreak(2,4),\allowbreak
(3,4),\allowbreak(4,5)\}$.
At the plug-in of DGU~6, the green links $(1,6)$ and $(5,6)$ are added.
Later, DGU~3 is unplugged, removing the red links $(1,3)$ and $(3,4)$.
The simulated sequence therefore consists of an initial load step, the
plug-in of DGU~6, a subsequent load step in the enlarged network, and
the unplugging of DGU~3.

To obtain the subsystem partition used in Section~\ref{sec:prelim},
each line-current state is assigned once to one of its incident DGUs.
The same physical line current remains available to both endpoint DGUs
when evaluating their local dynamics, but its state and reconstructed
residual appear in only one local block. Fixed slots are reserved for
all candidate line currents assigned to each DGU, with inactive slots
set to zero. Hence the dimensions of the surviving local models and
neural controllers remain unchanged across the considered PnP
configurations.

The simulations use the nonlinear ZIP-load model, while the IMC
controller uses a nominal model associated with each operating
configuration. We use $\mu$ to index a configuration, defined by the
active network topology and load parameters. For a fixed orientation
of each line, let $b_{i\ell}\in\{-1,0,1\}$ denote its incidence with
DGU $i$. The corresponding equilibrium satisfies
\begin{equation}\label{eq:local_equilibrium}
\begin{aligned}
\bar V_i
&=V_i^{\rm ref},\qquad
\bar I_\ell
=
\frac{1}{R_\ell}
\sum_j b_{j\ell}V_j^{\rm ref},\\
\bar I_{t,i}
&=
Y_{L,i}V_i^{\rm ref}
+\bar I_{L,i}
+\frac{P_{L,i}}{V_i^{\rm ref}}
+\sum_{\ell\in\mathcal E_i}b_{i\ell}\bar I_\ell,\\
\bar x_{I,i}
&=
\frac{
(1-k_{1,i})V_i^{\rm ref}
+(R_{t,i}-k_{2,i})\bar I_{t,i}}
{k_{3,i}} .
\end{aligned}
\end{equation}
Thus, for fixed voltage references, each DGU can determine the
equilibrium quantities required by its local model from local load
data, incident-line parameters, and neighboring voltage references.

For each configuration $\mu$, the nonlinear continuous-time model is
discretized using forward Euler with $T_{\rm samp}=5\cdot10^{-5}\,\mathrm{s}$.
The nominal IMC model is obtained by linearizing the same continuous-time
dynamics around~\eqref{eq:local_equilibrium} and applying the same
discretization.
Here $x_k$ is the physical state, while
$\delta x_k:=x_k-\bar x_{\mu_k}$ corresponds to
$x_t$ in Section~\ref{sec:prelim}.
Let $x_{k+1}^-$ denote the state propagated under
configuration $\mu_k$, immediately before a possible
PnP transition.
For the active configuration $\mu_k$, the resulting nonlinear
discrete-time model admits the decomposition
\begin{equation*}
x_{k+1}^-
=
\bar x_{\mu_k}
+
A_{\mu_k}\delta x_k
+
B_{\mu_k}\Delta u_k
+
r_{\mu_k}(\delta x_k),
\end{equation*}
where $A_{\mu_k}$ and $B_{\mu_k}$ define the nominal linear IMC model
and $r_{\mu_k}$ is the residual nonlinear model mismatch.

A PnP transition is represented by $x_{k+1}=\Gamma_k x_{k+1}^-+\sigma_k$, where $\Gamma_k:
\mathbb R^{n_{\mu_k}}
\to
\mathbb R^{n_{\mu_{k+1}}}$ maps the surviving state components from the
pre-change to the post-change configuration, while $\sigma_k$
initializes newly activated components.
For a load change without topology modification,
$n_{\mu_{k+1}}=n_{\mu_k}$, $\Gamma_k=I$, and $\sigma_k=0$. Expressing the resulting state relative to the
equilibrium of the new configuration gives
\begin{equation}\label{eq:num_deviation_model}
\delta x_{k+1}
=
\Gamma_k\!\left(
A_{\mu_k}\delta x_k
+
B_{\mu_k}\Delta u_k
+
r_{\mu_k}(\delta x_k)
\right)
+d_k,
\end{equation}
where $d_k
:=
\Gamma_k\bar x_{\mu_k}
+\sigma_k
-\bar x_{\mu_{k+1}} .$ 
At a PnP event with $t_{\mathrm p}=k+1$, the
post-change initial state is
$x_{t_{\mathrm p}}^+=\delta x_{k+1}$.
Using
$\Gamma_k(A_{\mu_k}\delta x_k+B_{\mu_k}\Delta u_k)$
as the nominal IMC prediction, the reconstructed residual is therefore
\[
\hat e_{k+1}
=
\Gamma_k r_{\mu_k}(\delta x_k)+d_k .
\]
Thus, the nonlinear model mismatch and the impulse generated by
equilibrium recentering and the PnP state transition enter the same
disturbance-feedback channel.

\subsection{Controller design and closed-loop results}

We first compare the two parametrizations of Section~\ref{sec:param}
for the controller interconnection used in the simulations. Each local
operator has one scalar output $z_i$, while its inputs consist of the
residual components assigned to DGU $i$ and one additional communication
channel. The matrix $M_{ve}$ injects each reconstructed-residual
component, with unit gain, into exactly one local operator according to
the state partition described above. The matrix $M_{vz}$ feeds the
communication channel of each DGU with the sum of the outputs of its
active neighbors, and $M_{uz}=I$ maps $z_i$ to the corresponding
performance-boosting channel.

For the six-DGU configuration after the plug-in of DGU~6 and before
the unplugging of DGU~3, let $A_{\cG}$ denote the binary adjacency
matrix of the active DGU graph. With unit communication weight,
$\gamma_{\mathbf R}=1$, and $\varepsilon\downarrow0$, maximizing over
the free variables gives
\[
\bar\gamma_{\rm nb}
=
\frac{1}{\sqrt{d_{\max}(\cG)(1+d_{\max}(\cG))}},
\qquad
\bar\gamma_{\rm norm}
=
\frac{1}{1+\|A_{\cG}\|_2}.
\]
For this topology, $d_{\max}(\cG)=3$ and
$\|A_{\cG}\|_2=2.39$, yielding
$\bar\gamma_{\rm nb}=0.289$ and
$\bar\gamma_{\rm norm}=0.295$.
The network-norm parametrization therefore provides, in this normalized setting with $\gamma_{\mathbf R}=1$, only a modest
increase in the largest certified common local gain for this
configuration. We use the neighbor-based parametrization below because
it additionally enables localized PnP recertification.

For the closed-loop simulations, we set $\gamma_{\mathbf R}=25$. For a configuration with $N_\mu$ active
DGUs, let $n_i$ denote the number of reconstructed-residual components
assigned to DGU $i$. Each local operator has input dimension
$q_i=n_i+1$, where the first $n_i$ channels receive the local residual
block and the last channel is used for communication, and has one
scalar output $z_i$ (see Fig.~\ref{fig:interconnectionREN}). Hence, $q=\sum_{i\in\mathcal V_\mu}q_i,$ $r=m=N_\mu$.
The residual-injection matrix is
$M_{ve}
=
\blkdiag_{i\in\mathcal V_\mu}
\left(
[
I_{n_i}, \ 
0_{1\times n_i}
]^\top
\right)
\in\mathbb R^{q\times n}$,
so that each residual component is injected with unit gain into exactly
one local operator. The matrix $M_{vz}\in\mathbb R^{q\times N_\mu}$
has nonzero entries only in the communication row of each local block, while $M_{uz}=I_{N_\mu}$.

Each local operator $\mathbf R^{[i]}$ is implemented as a
REN~\cite{revay2023recurrent} satisfying
Assumption~\ref{ass:local_diss}, with input dimension $q_i$, scalar
output, $6$ internal states, and equilibrium-network dimension $6$.
Its local gain $\gamma^{[i]}$ is determined by the neighbor-based free
parametrization for the current configuration.
The outputs $z_i=(\mathbf{R}^{[i]}(\mathbf{v}^{[i]}))_k$ of the interconnected operator are used
both in the neighbor-to-neighbor interconnection through $M_{vz}$ and
as the magnitude signals for a local Magnitude And Direction (MAD)
policy~\cite{furieri2025mad,cao2025scaling}. Specifically,
\begin{equation}\label{eq:mad_factorization}
\begin{aligned}
\Delta u_i(k)
&=
|(\mathbf{R}^{[i]}(\mathbf{v}^{[i]}))_k|
\tanh \ (\mathrm{MLP}_i(\psi_i(k))),
\end{aligned}
\end{equation}
where $\psi_i(k)$ collects local and neighbor-available voltage-tracking
features. The direction term is applied after the certified
interconnection: the signals exchanged through $M_{vz}$ are the
pre-MAD outputs $z_i$. Since $|\tanh (\mathrm{MLP}_i(\psi_i(k)))| \le1$,
$\|\Delta\uvec\|
\le
\|\mathbf z\|
\le
\gamma_{\mathbf R}\|\hat{\evec}\|+\beta_R,$
so the direction networks do not increase the prescribed network gain.
Each direction map is represented by a local MLP with one hidden-layer of width~8.

For rollout $s$, let
$e_{V,i}^{(s)}(k):=V_i^{(s)}(k)-V_i^{\rm ref}$.
The training objective is written as
\begin{equation}\label{eq:training_loss_num}
J^{(s)}
=
\sum_{k=0}^{T-1}
\left(
\ell_{\rm reg}^{(s)}(k)
+
\ell_{\rm pen}^{(s)}(k)
\right),
\end{equation}
where $\ell_{\rm reg}$ collects quadratic penalties on voltage-tracking
errors, the performance-boosting inputs $\Delta u$, and their time increments. The term $\ell_{\rm pen}$ collects additional penalties
used around load and PnP events, including voltage excursions, peak
errors, post-event settling, and deviations outside the operating tube
$|V_i-V_i^{\rm ref}|\le 10\,\mathrm V$. In particular, the latter is
penalized through terms proportional to
$\bigl(|e_{V,i}^{(s)}(k)|-10\,\mathrm V\bigr)_+^2$, where
$(a)_+:=\max\{a,0\}$.
The complete loss and numerical weights are provided in the
accompanying code.

Training is performed centrally for computational convenience using the
unconstrained formulation of Section~\ref{sec:disttrain}, with horizon
$T=3000$. The DGU-specific REN and MAD parameters are trained jointly over several
load changes and admissible PnP configurations, and the same parameters
for each DGU are used across all configurations in which that DGU is
present. Candidate plug-in units are
included during this training. 

For the robust check, however, global gain bounds are not available
because of the nonlinear ZIP loads. We therefore assess the condition of
Proposition~\ref{prop:robust_imc} on the operating tube
$|V_i-V_i^{\rm ref}|\le10\,\mathrm V$, whose violation is penalized
during training and which is verified a posteriori for all simulated
trajectories. Within this region, the pre-stabilized dynamics admit a Lur'e
representation, and the ZIP nonlinear remainder is sector bounded.
A standard dissipativity LMI (cf.~\cite{arcak2016networks}) provides the regional gain
$\gamma_{\F}^\mu$ from $\Delta\uvec$ to $\x$. The same sector bound
gives $\gamma_{\boldsymbol\Delta}^\mu$.
Across the considered configurations, the largest bounds are
$\gamma_{\F}^\mu=14.26$ and
$\gamma_{\boldsymbol\Delta}^\mu=2.2\cdot10^{-5}$.
With $\gamma_{\mathbf R}=10$, the regional small-gain check
$\gamma_{\boldsymbol\Delta}^\mu
(\gamma_{\F}^\mu+1)\gamma_{\mathbf R}=8.5\cdot10^{-3}<1$
is satisfied for every configuration. The PnP recentering term $d_k$ is treated
separately as an impulse $\ell_2$ residual and does not enter the mismatch gain.

Figures~\ref{fig:Plugin}--\ref{fig:loadthm1} compare the proposed
controller with the decentralized primary baseline.
Figure~\ref{fig:Plugin} shows the plug-in of DGU~6 at
$t=4\,\mathrm{s}$, Fig.~\ref{fig:loadthm1} the subsequent load
increase at $t=8\,\mathrm{s}$, and Fig.~\ref{fig:Plugout} the
unplugging of DGU~3 at $t=12\,\mathrm{s}$.
Quantitatively, over the $0.4\,\mathrm{s}$ following each event
($N_{\rm eval}=8000$), the proposed controller reduces the peak voltage
deviation by $60.1\%$, $59.1\%$, and $56.4\%$ for the plug-in, load
change, and unplugging, respectively, while the corresponding reductions
in integrated squared voltage error are $15.3\%$, $76.1\%$, and $72.4\%$.

\begin{figure*}
    \centering
    \includegraphics[width=0.9\linewidth]{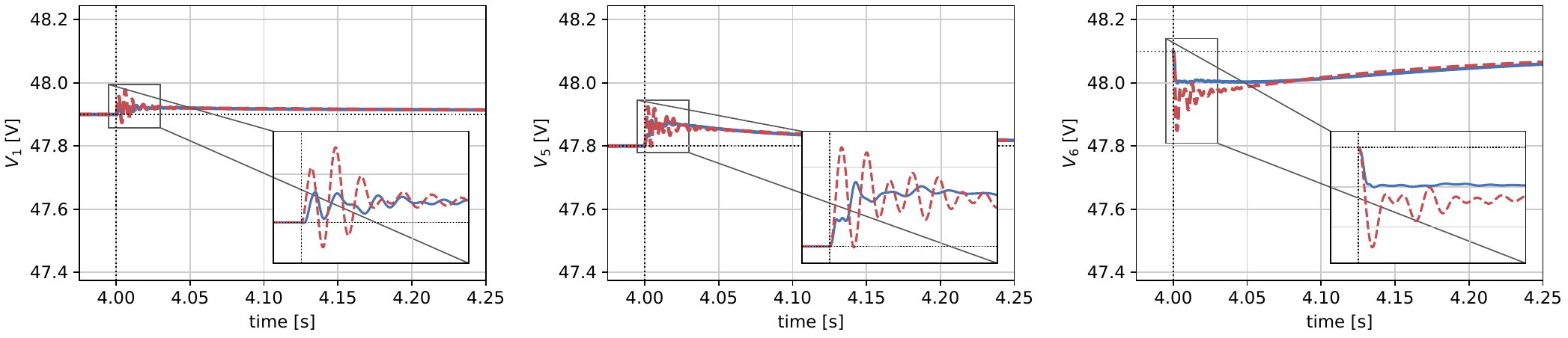}
    \caption{Plug-in of DGU~6 at $t=4\,\mathrm{s}$: PCC voltages under the
proposed controller (blue) and the decentralized baseline~\cite{nahata2020passivity} (red dashed).}
    \label{fig:Plugin}
\end{figure*}

\begin{figure*}
    \centering
    \includegraphics[width=0.6\linewidth]{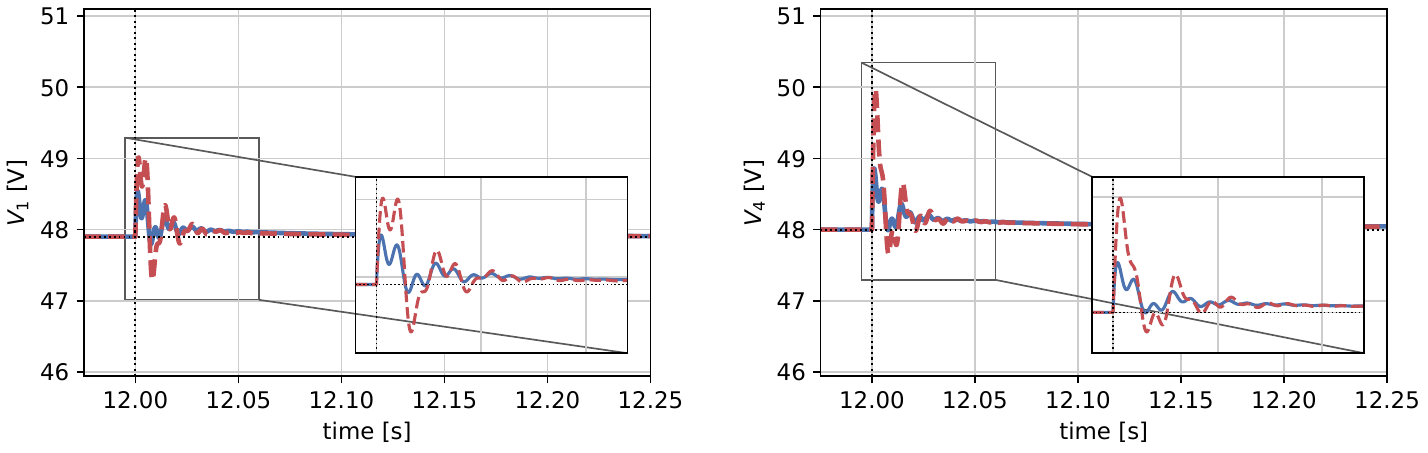}
    \caption{Unplugging of DGU~3 at $t=12\,\mathrm{s}$ under the proposed
controller (blue) and the decentralized baseline~\cite{nahata2020passivity} (red dashed).}
    \label{fig:Plugout}
\end{figure*}

\begin{figure*}
    \centering
    \includegraphics[width=0.9\linewidth]{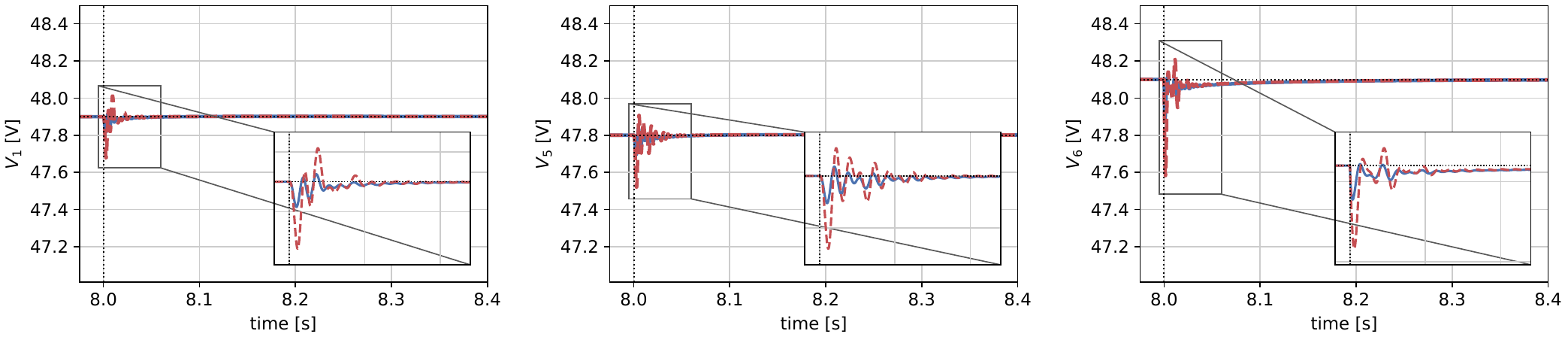}
    \caption{Load disturbance at PCC~6 at $t=8\,\mathrm{s}$ in the enlarged network under the proposed
controller (blue) and the decentralized baseline~\cite{nahata2020passivity} (red dashed).}
    \label{fig:loadthm1}
\end{figure*}

\section{Conclusions}\label{sec:conclusions}

We developed a framework for fail-safe training of distributed neural
performance-boosting controllers for pre-stabilized PnP networks. An IMC
reformulation reduces closed-loop stability to a prescribed network
$\ell_2$ gain, enforced through two free parametrizations. The
neighbor-based construction enables localized PnP recertification, whereas
the network-norm construction uses network-wide coupling information and
can reduce conservatism. For strictly causal local realizations and separable costs, training
admits a neighbor-to-neighbor implementation while every iterate preserves
the prescribed gain. A DC-microgrid study showed improved voltage
transients on top of a decentralized PnP-stable baseline.

\bibliographystyle{IEEEtran}
\bibliography{biblio}

\appendix

\section{Proof of Proposition~\ref{prop:robust_imc}}
\label{app:proof_robust_imc}

We use a truncation argument to avoid assuming that the closed-loop
signals are in $\ell_2$ a priori. For a sequence $\mathbf a$, let
$\|\mathbf a\|_T
:=
\left(\sum_{t=0}^T\|a_t\|^2\right)^{1/2}.$
By causality, the bounds
\eqref{eq:true_plant_gain}--\eqref{eq:mismatch_gain} and the gain bound
on $\mathbf R$ hold on every truncated horizon. From \eqref{eq:model_mismatch},
$\hat{\evec}
=
\x-\widehat\F(\x,\uvec)
=
\boldsymbol\Delta(\x,\uvec)
+\mathbf E(x_0)+\mathbf S\w $.
Hence,
\begin{align}
\|\hat{\evec}\|_T
&\le
\gamma_{\boldsymbol\Delta}
\bigl(\|\x\|_T+\|\uvec\|_T\bigr)
+\|x_0\|+\|\w\|_T \label{eq:proof_ehat_pre}\\
&\le
(1+\gamma_{\boldsymbol\Delta}\gamma_{\F})\|\w\|_T
+\gamma_{\boldsymbol\Delta}(1+\gamma_{\F})\|\uvec\|_T
+b_0,
\notag
\end{align}
where
$b_0:=\|x_0\|+\gamma_{\boldsymbol\Delta}\beta_F(\|x_0\|)$.
Using
$\|\uvec\|_T\le
\gamma_{\mathbf R}\|\hat{\evec}\|_T+\beta_R$
in \eqref{eq:proof_ehat_pre} gives
\[
\Gamma\|\hat{\evec}\|_T
\le
(1+\gamma_{\boldsymbol\Delta}\gamma_{\F})\|\w\|_T
+b_1,
\]
where
$\Gamma
:=
1-\gamma_{\boldsymbol\Delta}(1+\gamma_{\F})
\gamma_{\mathbf R},
$ $
b_1
:=
b_0+
\gamma_{\boldsymbol\Delta}(1+\gamma_{\F})\beta_R .$
By \eqref{eq:robust_imc_condition}, $\Gamma>0$, and therefore
\[
\|\hat{\evec}\|_T
\le
\frac{1+\gamma_{\boldsymbol\Delta}\gamma_{\F}}{\Gamma}
\|\w\|_T
+
\frac{b_1}{\Gamma}.
\]
Substitution into the bounds for $\mathbf R$ and the true plant then
gives
\[
\|\uvec\|_T\le c_u\|\w\|_T+b_u,
\qquad
\|\x\|_T\le c_x\|\w\|_T+b_x,
\]
for finite constants $c_u,c_x\ge0$ and
$b_u,b_x\ge0$ independent of $T$ (with the additive terms possibly
depending on $x_0$). Since $\w\in\ell_2$, these bounds are uniform in
$T$. Taking $T\to\infty$ proves that
$\uvec\in\ell_2^m$ and $\x\in\ell_2^n$, and hence the closed loop is
finite-gain $\ell_2$-stable.
\hfill$\blacksquare$

\section{Proof of Theorem~\ref{thm:param_local}}
\label{app:proof_param_local}

Let
$c^{[i]}
:=
\alpha^{[i]}(\gamma^{[i]})^2.$
With the notation of Proposition~\ref{prop:stabL2}, condition
$G^\top QG\preceq0$ is equivalent to
\begin{equation}\label{eq:LMI_extended_app}
\begin{bmatrix}
\Pi_{N,z}-H-M_{vz}^\top\Pi_{N,v}M_{vz}
&
-M_{vz}^\top\Pi_{N,v}M_{ve}
\\
-M_{ve}^\top\Pi_{N,v}M_{vz}
&
\gamma_{\mathbf R}^2I_n
-M_{ve}^\top\Pi_{N,v}M_{ve}
\end{bmatrix}
\succeq0 .
\end{equation}

By Assumption~\ref{ass:MzMe}, $M_{ve}$ is a permuted coordinate
injection. Hence $M_{ve}^\top\Pi_{N,v}M_{ve}$ is diagonal and selects
the entries $c^{[i]}$ associated with the coordinates in
$\mathcal J_{v_i}^1$. From
\eqref{eq:Bi_local}--\eqref{eq:c_local_compact}, for every subsystem $i$ with
$\mathcal J_{v_i}^1\neq\emptyset$,
\[
c^{[i]}
<
\bar c^{[i]}
\le
\frac{\gamma_{\mathbf R}^2}
{1+\rho_{v,1}^{[i]}\gamma_{\mathbf R}^2}
\le
\gamma_{\mathbf R}^2.
\]
Therefore, $\gamma_{\mathbf R}^2I_n
-
M_{ve}^\top\Pi_{N,v}M_{ve}
\succ0.$\\
Taking the Schur complement of
\eqref{eq:LMI_extended_app}, it is sufficient to show
\begin{equation}\label{eq:schur_local_app}
\Pi_{N,z}-H-M_{vz}^\top\widetilde D M_{vz}
\succeq0,
\end{equation}
where
\begin{multline*}
   \widetilde D
:=
\Pi_{N,v}
+ \\
\Pi_{N,v}M_{ve}
\left(
\gamma_{\mathbf R}^2I_n
-
M_{ve}^\top\Pi_{N,v}M_{ve}
\right)^{-1}
M_{ve}^\top\Pi_{N,v}.
\end{multline*}
Because $M_{ve}$ is a coordinate injection,
$\widetilde D$ is diagonal. For a coordinate $k$ associated with the local operator
$\mathbf R^{[i]}$,
$\widetilde D_{kk}
=
c^{[i]},$ if $k\in\mathcal J_{v_i}^0$, or
$\widetilde D_{kk}
=
\frac{c^{[i]}\gamma_{\mathbf R}^2}
{\gamma_{\mathbf R}^2-c^{[i]}},
$ when $ k\in\mathcal J_{v_i}^1$.
Moreover, $\widetilde D\succ0$. Hence, by another Schur complement,
\eqref{eq:schur_local_app} is equivalent to
\begin{equation}\label{eq:schur_gershgorin_app}
\begin{bmatrix}
\Pi_{N,z}-H & M_{vz}^\top\\
M_{vz} & \widetilde D^{-1}
\end{bmatrix}
\succeq0.
\end{equation}

We now verify \eqref{eq:schur_gershgorin_app} by diagonal dominance.
Consider first a row associated with
$z_j$, $j\in\mathcal J_{z_i}$. Its diagonal entry satisfies
$\alpha^{[i]}-h_j
\ge
\rho_z^{[i]}
+(\eta_\alpha^{[i]})^2+\varepsilon
>
\rho_z^{[i]}$,
where we used $h_j\le\bar h_i$. By
\eqref{eq:rho_v0_z_local}, $\rho_z^{[i]}$ bounds the corresponding
absolute column sum of $M_{vz}$ and therefore the sum of the absolute
values of the off-diagonal entries in this row. Consider now a row associated with
$v_k$, and define
$r_k:=\sum_j |(M_{vz})_{kj}|.$
If $k\in\mathcal J_{v_i}^0$, then
$r_k\le\rho_{v,0}^{[i]}$. If $\rho_{v,0}^{[i]}=0$, the row has no
off-diagonal contribution from $M_{vz}$. Otherwise,
\eqref{eq:Bi_local}--\eqref{eq:c_local_compact} give
$c^{[i]}
<
\bar c^{[i]}
\le
1/\rho_{v,0}^{[i]}$,
and therefore $\widetilde D_{kk}r_k
=
c^{[i]}r_k
<1.$
If $k\in\mathcal J_{v_i}^1$, then
$r_k\le\rho_{v,1}^{[i]}$ and
\[
c^{[i]}
<
\frac{\gamma_{\mathbf R}^2}
{1+\rho_{v,1}^{[i]}\gamma_{\mathbf R}^2}.
\]
Rearranging this inequality gives
$\frac{c^{[i]}\gamma_{\mathbf R}^2}
{\gamma_{\mathbf R}^2-c^{[i]}}
\rho_{v,1}^{[i]}
<1.$
Hence, $\widetilde D_{kk}r_k<1.$
Thus, for every row of
\eqref{eq:schur_gershgorin_app}, the diagonal entry is not smaller
than the sum of the absolute values of the off-diagonal entries.
Since the matrix is symmetric with nonnegative diagonal entries, it is
positive semidefinite. Therefore
\eqref{eq:conditionL22} holds.
\hfill$\blacksquare$

\section{Proof of Theorem~\ref{thm:param_spectral}}
\label{app:proof_param_spectral}

By construction, $\alpha^{[i]}(\gamma^{[i]})^2=c,
 \ \forall i\in\cV,$
and therefore $\Pi_{N,v}=cI_q.$
Moreover, \eqref{eq:c_spectral} gives $0<c<\gamma_{\mathbf R}^2$, $\forall \eta_{\rm{net}}\in\mathbb R$.
Using $M_{ve}^\top M_{ve}=I_n$, the lower-right block of
\eqref{eq:LMI_extended_app} becomes
\[
\gamma_{\mathbf R}^2I_n
-
M_{ve}^\top\Pi_{N,v}M_{ve}
=
(\gamma_{\mathbf R}^2-c)I_n
\succ0.
\]
Taking the Schur complement, it is sufficient to show
\begin{equation}\label{eq:schur_spectral_app}
\Pi_{N,z}-H-M_{vz}^\top\widetilde D M_{vz}\succeq0,
\end{equation}
where
$\widetilde D
=
cI_q
+
\frac{c^2}{\gamma_{\mathbf R}^2-c}
M_{ve}M_{ve}^\top$.
Since $M_{ve}^\top M_{ve}=I_n$,
$M_{ve}M_{ve}^\top$ is an orthogonal projector and hence
$M_{ve}M_{ve}^\top\preceq I_q$. Therefore,
\[
\widetilde D
\preceq
\left(
c+\frac{c^2}{\gamma_{\mathbf R}^2-c}
\right)I_q
=
\tau I_q.
\]
It follows that $M_{vz}^\top\widetilde D M_{vz}
\preceq
\tau M_{vz}^\top M_{vz}
\preceq
\tau \widehat\kappa I_r.$ 
Thus, \eqref{eq:schur_spectral_app} is satisfied if
$\Pi_{N,z}-H-\tau\widehat\kappa I_r\succeq0$.
For every coordinate $j\in\mathcal J_{z_i}$,
\[
\alpha^{[i]}-h_j-\tau\widehat\kappa
\ge
(\eta_\alpha^{[i]})^2+\varepsilon
>0,
\]
where we used $h_j\le\bar h_i$. Hence $\Pi_{N,z}-H-\tau\widehat\kappa I_r
\succ0,$
which proves~\eqref{eq:conditionL22}.
\hfill$\blacksquare$

\section{Proof of Theorem~\ref{thm:pnp_locality_new}}
\label{app:proof_pnp_locality}

Consider first a surviving subsystem $i\notin\mathcal I$. By
definition of~$\mathcal I$, $\bar h_i^+ + \rho_z^{+[i]}
=
\bar h_i + \rho_z^{[i]}$, $\bar c^{+[i]}=\bar c^{[i]}$.
Since the recertification procedure also keeps
$\eta^{+[i]}=\eta^{[i]}$,
\eqref{eq:alpha_local_compact}--\eqref{eq:gamma_local_compact}
give
$\alpha^{+[i]}=\alpha^{[i]}$,
$c^{+[i]}=c^{[i]}$, and
$\gamma^{+[i]}=\gamma^{[i]}$.
\\
For every $i\in\mathcal I$, the quantities
$(\alpha^{+[i]},\gamma^{+[i]})$ are instead constructed from the
post-change quantities according to
Theorem~\ref{thm:param_local}. Therefore the complete collection $\{(\alpha^{+[i]},\gamma^{+[i]})\}_{i\in\cV^+}$ satisfies the neighbor-based parametrization of
Theorem~\ref{thm:param_local} for the post-change interconnection.
Hence~\eqref{eq:conditionL22} holds with the same
$\gamma_{\mathbf R}$.
\\
It remains to verify the PnP update constraints. Consider a surviving
subsystem $i\notin\mathcal A(\cG,\cG^+)$. By
Definition~\ref{def:pnp_update},
$\cN_i^+=\cN_i$ and $\F^{+[i]}\equiv\F^{[i]}$.
Since the neighborhood of $i$ is unchanged, no network connection
incident to $i$ is added or removed. Under the interconnection update
described in Section~\ref{sec:pnp}, the coupling blocks of $M_{vz}$
associated with subsystem $i$ are therefore unchanged. The corresponding
blocks of $M_{ve}$ and $M_{uz}$ are also unchanged by construction.
Consequently,
$\rho_z^{+[i]}=\rho_z^{[i]}$,
$\rho_{v,0}^{+[i]}=\rho_{v,0}^{[i]}$,
$\rho_{v,1}^{+[i]}=\rho_{v,1}^{[i]}$, and
$\bar h_i^+=\bar h_i$.
The sets entering~\eqref{eq:Bi_local} are therefore unchanged and,
since $\bar c_0^{[i]}$ is fixed across network configurations,
\eqref{eq:cbar_local} gives
$\bar c^{+[i]}=\bar c^{[i]}$.
Thus $i\notin\mathcal I$, which proves
$\mathcal I\subseteq\mathcal A(\cG,\cG^+)$. The recertification procedure therefore leaves
$\theta^{[i]}$, $\eta^{[i]}$, $\alpha^{[i]}$, and
$\gamma^{[i]}$ unchanged for every surviving subsystem
$i\notin\mathcal A(\cG,\cG^+)$. Together with
$\F^{+[i]}\equiv\F^{[i]}$, $\cN_i^+=\cN_i$, and the unchanged
controller-interconnection blocks, this implies that its complete
local controller realization is unchanged, namely
$\K^{+[i]}\equiv\K^{[i]}$ for all
$i\in(\cV\cap\cV^+)\setminus\mathcal A(\cG,\cG^+)$.
Hence~\eqref{eq:pnp_update_constraint} holds.
\hfill$\blacksquare$

\end{document}